\documentclass[a4paper,UKenglish,cleveref, autoref, thm-restate]{lipics-v2021}
\usepackage{standalone}
\usepackage{mathtools}
\usepackage{multirow}
\usepackage{amsmath,amsfonts,amsthm}
\usepackage{tikz}
\usepackage{comment}
\usetikzlibrary{intersections}
\usetikzlibrary{calc}
\usetikzlibrary{arrows}
\usepackage{dsfont}
\usepackage[ruled,vlined, linesnumbered]{algorithm2e}

\SetCommentSty{mycommfont}
\SetAlFnt{\small}
\newtheorem{prop}{Proposition}

\SetKwComment{Comment}{/* }{ */}

\title{A Real-time Calculus Approach for Integrating Sporadic Events in Time-triggered Systems}

\author{Ana\"{i}s Finzi}{TTTech Computertechnik AG, Austria}{anais.finzi@tttech.com}{}{}

\author{Silviu S. Craciunas}{TTTech Computertechnik AG, Austria}{silviu.craciunas@tttech.com}{}{}

\author{Marc Boyer}{ONERA / DTIS, Universit\`{e} de Toulouse, France }{marc.boyer@onera.fr}{orcid.org/0000-0003-0344-6991}{}

\authorrunning{A. Finzi, S.S. Craciunas, M. Boyer} 
\titlerunning{A Real-time Calculus Approach for Integrating Sporadic Events in TT Systems}
\Copyright{Ana\"{i}s Finzi, Silviu S. Craciunas, Marc Boyer} 

\ccsdesc[500]{Computer systems organization~Dependable and fault-tolerant systems and networks}

\keywords{time-triggered, event-triggered, scheduling, real-time, real time calculus}

\newcommand{\nn}[1]{\left[#1\right]^+}            
\newcommand{\nnup}[1]{\left[#1\right]^+_\uparrow} 

\newcommand{\hDev}{\textit{hDev}}

\newcommand{\set}[1]{\left\{{#1}\right\}}
\newcommand{\isdef}{\ensuremath{\overset{\textit{def}}{=}}}

\newcommand{\card}[1]{\left|{#1}\right|}

\newcommand{\nnR}{\ensuremath{\mathbb{R}^+}}

\usepackage{xspace}
\newcommand{\ie}{\emph{i.e.}\xspace}

\newcommand{\mt}{\textit{mt}}

\hideLIPIcs
\begin{document}
\nolinenumbers
\maketitle

\begin{abstract}
In time-triggered systems, where the schedule table is predefined and statically configured at design time, sporadic event-triggered (ET) tasks are handled within specially dedicated slots or when time-triggered (TT) tasks finish their execution early. We introduce a new paradigm for synthesizing TT schedules that guarantee the correct temporal behavior of TT tasks and the schedulability of sporadic ET tasks with arbitrary deadlines. The approach first expresses a constraint for the TT task schedule in the form of a maximal affine envelope that guarantees that as long as the schedule generation respects this envelope, all sporadic ET tasks meet their deadline. The second step consists of modeling this envelope as a burst limiting constraint and building the TT schedule via simulating a modified Least-Laxity-First (LLF) scheduler. Using this novel technique, we show that we achieve equal or better schedulability and a faster schedule generation for most use-cases compared to simple polling approaches. Moreover, we present an extension to our method that finds the most favourable schedule for TT tasks with respect to ET schedulability, thus increasing the probability of the computed TT schedule remaining feasible when ET tasks are later added or changed.
\end{abstract}

\section{Introduction}
\label{sec:introduction}
Time-triggered systems have been used to a great extent in the aerospace domain where the safety-critical nature of the applications imposes a certain level of determinism on the architecture, especially when certification is required~\cite{10.1145/2997465.2997492, DEUTSCHBEIN2019102}. Moreover, the automotive sector has recently seen a push towards centralizing functionality onto a more scalable and flexible integrated platform (c.f.~\cite{niedrist}) in order to support the complex real-time needs of, e.g., ADAS subsystems~\cite{Fleming2015, McLeanFRONTIERS22} and to allow a mixed-criticality paradigm. Thus, the use of time-triggered scheduling (cyclic executive) solutions leading to more deterministic (and thus more easily verifiable and certifiable) systems is also gaining importance in the automotive domain~\cite{6165039, Sagstetter:2014:SIF:2593069.2593211, ernst_et_al:DR:2018:9293}. In particular, the complex jitter and multi-rate cause-effect requirements found in ADAS applications~\cite{7579951, Becker:2017:ETA:3165725.3165925} cannot be easily guaranteed off-line using classical fixed- or dynamic-priority approaches and necessitate a more predictable time-triggered architecture (TTA)~\cite{McLeanFRONTIERS22}. While TTA has many benefits in terms of predictability, stability, compositionality, and determinism, the use of a static schedule table is notoriously inefficient at integrating sporadic event-driven tasks (ET). Conversely, pure event-triggered systems suffer from many drawbacks compared to a time-triggered execution, e.g., high jitter and starvation (c.f~\cite{10.1016/j.sysarc.2019.101652, 10.1023/A:1008198310125, 10.1007/BF00365463, 6823162, isovic09}). Modern safety-critical systems benefit most from combining the two paradigms, allowing a time-triggered system to be flexible enough to respond to sporadic events when needed. 

For time-triggered systems, where the schedule table is predefined and statically configured at design time, sporadic event-triggered (ET) tasks can only be handled within specially allocated slots or when time-triggered (TT) tasks finish their execution earlier than their worst-case assumption. While there is a significant body of work (c.f.\cite{10.1145/3431232} for an extensive survey) concerning pure time-triggered schedule generation, which is an NP-complete problem, most of the methods do not consider the schedulability of sporadic ET tasks. Traditionally, the integration of sporadic ET tasks in time-triggered systems is either done via a feedback loop integrated into the TT schedule generation mechanism~\cite{tpop2002, tpop2003}, or via hierarchical scheduling~\cite{10.1145/1017753.1017772, Shin2003, Shin2008, 1019197}. For both approaches, the computational effort (on top of the complexity of creating TT schedules) can be significant due to the response time analysis for each variation of TT slot placement or due to solving the server design problem within the TT schedulability space. Therefore, the challenge is to create static schedule tables for which both TT and ET tasks respect their deadlines while keeping the computational effort low.

We present a novel approach in which we first compute a maximal affine envelope (defined by a maximum burst and a rate) for the TT tasks in the system, such that as long as a TT schedule respects this envelope, all sporadic ET tasks meet their deadline. The second step involves expressing this envelope as a burst limiting constraint on the TT schedule and building the static schedule table via simulating a modified Least-Laxity-First (LLF) scheduler. Using this novel technique, we trade-off complexity for exactness via the pessimism of the affine envelope approximation resulting in a faster schedule generation while still achieving equal or better schedulability compared to the simple polling approach. Moreover, this method enables an efficient design optimization technique for iterative design processes where ET tasks are added or changed later. Our contributions are, therefore:
\begin{itemize}
\item a new and efficient approach based on affine envelope approximations for guaranteeing the schedulability of ET tasks with arbitrary deadlines without the need for complex response-time analysis,
\item a novel LLF-based algorithm that respects the affine envelope (expressed as a burst limiting constraint) and produces a static schedule guaranteeing both TT and ET deadlines,
\item a computationally ``cheap'' method for integrating ET tasks in TT systems that, while being pessimistic for some task sets, has in most cases equal or better schedulability and runtime results compared to other methods based on the traditional polling approach,
\item a design optimization where we maximize the solution space for changing or adding ET tasks without modifying the existing TT schedule generated via our method.
\end{itemize}

We start by introducing some necessary preliminaries in Section~\ref{sec:preliminaries}, followed by a review of related literature in Section~\ref{sec:related_work}. We introduce two polling approaches that follow earlier results in Section~\ref{sec:polling_approach} and our novel method based on affine envelope approximations in Section~\ref{sec:tt_envelope}. After evaluating our method in Section~\ref{sec:experiments} we draw some conclusions in Section~\ref{sec:conclusion}. 

\section{Preliminaries}
\label{sec:preliminaries}
\subsection{System model}
\label{sec:model}

\begin{figure}[!ht]
    \centering
	\includegraphics[width=0.99\textwidth]{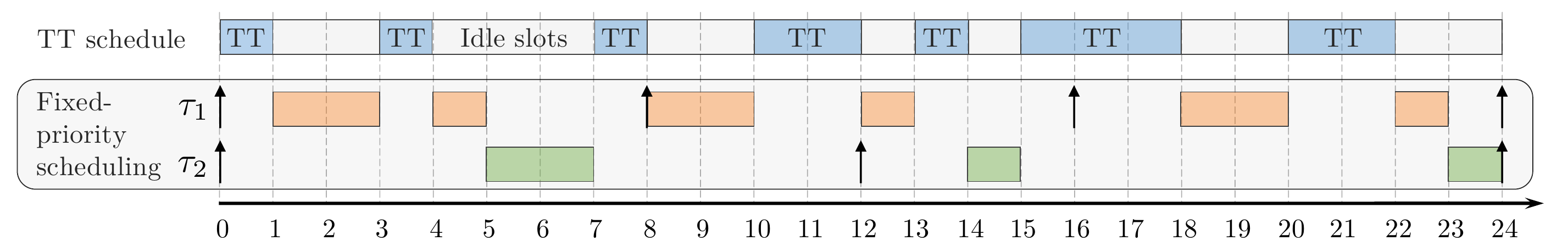}
	\caption{\label{fig:system} Static schedule table with $2^{nd}$-level fixed-priority scheduling in idle slots.}
\end{figure}

We assume a task dispatcher that schedules TT tasks based on an offline generated static schedule table (cyclic executive), and in the slots that are left for ET tasks, implements a $2^{nd}$-level preemptive scheduler based on fixed priorities (c.f. Figure~\ref{fig:system}). We denote the set of TT and ET tasks with $\mathcal{T}^{TT}$ and $\mathcal{T}^{ET}$, respectively. 

A TT or ET task $\tau_i$ is defined by the tuple ($p_i$, $C_i$, $T_i$, $D_i$) with $C_i$ denoting the computation time, $p_i$ is the task priority, and $D_i$ the relative deadline of the task. For TT tasks, $T_i$ represents the period, while for ET tasks where we assume a sporadic model, it describes the minimal inter-arrival distance (MIT). Usually, TT tasks have a constrained-deadline model ($T_i \leqslant D_i$), while ET tasks can have an arbitrary deadline, i.e., it can also be larger than the inter-arrival time. For convenience, we say that all TT tasks share the same (highest) priority. Event-triggered tasks, having a lower priority than TT tasks, are indexed in the order of their relative priority, i.e., $\tau_i$ has higher priority than $\tau_j$ ($p_i > p_j$) implies $i > j$, but several tasks may have the same priority ($p_i=p_j$) in which case they are selected in FIFO order.

The timeline of scheduling decisions is divided into equal segments by the microtick $mt$, representing the smallest scheduling granularity for tasks~\cite{235008, isovic09}. Usually, the granularity of the timeline is in the range of hundreds of microseconds to a few milliseconds; however, we do not assume any lower value here as the granularity in, e.g., embedded devices with custom runtime systems can go down to the order of microseconds (e.g.~\cite{CraciunasETFA14}). In the following, we assume that any $D,C,T$ are multiples of \mt. A {schedule} for a finite set of tasks $\mathcal{T} = \mathcal{T}^{TT} \cup \mathcal{T}^{ET}$ is a partial function $\sigma: \mathbb{N} \hookrightarrow \mathcal{T}$ from the time domain to the set of tasks, that assigns to each interval $[t\cdot mt,(t+1)\cdot mt )$ defined by the microtick granularity a task that is running in that time interval. We assume that each schedule $\sigma$ repeats after a certain time period called the schedule cycle, which is usually equal to the hyperperiod of the system, defined by $T = \textit{lcm}_{\tau_i \in \mathcal{T}^{TT}}\set{T_i}$. Furthermore, we assume that the system is not overloaded, \ie $U \leqslant \lambda$, with $\lambda$ the computation capacity of the system and hence $\sigma(t)$ is uniquely defined for each point on the microtick timeline. We consider in this work that the tasks from both sets $\mathcal{T}^{TT}$ and $\mathcal{T}^{ET}$ are scheduled on a single-core CPU with capacity $\lambda=1$.

We introduce a few notations to ease readability. For any task $\tau_i$, $p(i)$ is the priority of the task, $U_i=\frac{C_i}{T_i}$ is the utilization of the task $\tau_i$,
$U^{TT}=\sum_{\tau_i \in \mathcal{T}^{TT}} U_i$ is the utilization of all TT tasks, 
$U^{ET}=\sum_{\tau_i \in \mathcal{T}^{ET}} U_i$ is the utilization of all ET tasks,
$U_{>p}=\sum_{\tau_i \in \mathcal{T}, p(i) > p} U_i$ is the utilization of all tasks with priority higher than $p$,
$U^{ET}_{=p}=\sum_{\tau_i \in \mathcal{T}^{ET}, p(i) = p} U_i$ is the utilization of all ET tasks with priority equal to $p$,
$U^{ET}_{>p}=\sum_{\tau_i \in \mathcal{T}^{ET}, p(i) > p} U_i$ is the utilization of all ET tasks with priority higher than $p$, and notice that if $p$ is the priority of an ET task, $U_{>p}=U^{TT} + U^{ET}_{>p}$ since the TT task have higher priority than ET tasks. Using the same pattern, we define $C^{TT}, C^{ET}_{>p},  C^{ET}_{=p}, C^{ET}_{\geqslant p}$ and notice that $C^{ET}_{>p}+  C^{ET}_{=p} = C^{ET}_{\geqslant p}$.

\subsection{Real-time (and network) calculus}
\label{sec:network_calculus}
Network Calculus (NC)~\cite{NC-Book} is a theory for quantifying worst-case (latency and backlog) bounds in computer networks, using min-plus and max-plus algebra to relate the minimum service of network nodes and the maximum amount of flow traffic. Real-time calculus (RTC)~\cite{thiele2000real} is the real-time systems equivalent of NC for modeling and analyzing the worst-case behavior of tasks (e.g.~\cite{SUDHAKAR2021101856}). It also uses non-decreasing functions to model the maximum task computation demand and the minimum available CPU computation service in any specified time interval. It can thus be seen as a variant of the classical NC framework with some minor differences~\cite{8864582}. We only introduce the most important definitions and refer the reader to~\cite{thiele2000real, moy:hal-00442257} for a more in-depth description. 

The response time (i.e., delay) of a task $\tau_i$ of a set of tasks $\mathcal{T}$ is detailed in Theorem~\ref{thm:delay-bounds}. It depends on i) the maximum amount of requested computation of $\mathcal{T}$, represented by a so-called arrival curve $\alpha(t)$ defined in Definition~\ref{eq:alpha-def},  and ii) the minimum computation capacity available to $\mathcal{T}$, represented by a so-called minimum service curve $\beta(t)$ defined in Definition~\ref{def:service-curves}. Additionally, we define $\gamma$ in the maximum service curve, which represents the maximum amount of computation available to $\mathcal{T}$. As introduced in~\cite{thiele2000real}, the request function $R(t)\geqslant 0$ constitutes the accrued computation time solicited until time $t$. Conversely, the function $\mathcal{C}(t)\geqslant 0$ is the 
maximum computation time delivered until time $t$~\cite{thiele2000real}.

\begin{definition}[Arrival curve~\cite{thiele2000real}]
The arrival curve $\alpha (t)$ of a request function $R(t)$ is a non-decreasing function which satisfies:
\begin{equation}
 \label{eq:alpha-def}  
  R(t) - R(s) \leqslant \alpha(t-s), \forall s\leqslant t.
\end{equation}
\end{definition}

\begin{definition}[Service curves~\cite{thiele2000real, chakraborty2003general, 8864582}] \label{def:service-curves}
The maximum service curve $\gamma(t)$ and minimum service curve $\beta(t)$ of a capacity function $\mathcal{C}(t)$ are non-negative and non-decreasing functions satisfying:
\begin{equation}
\beta(t-s) \leqslant  \mathcal{C}(t)-\mathcal{C}(s) \leqslant  \gamma(t-s), \forall s\leqslant t.
\end{equation}
\end{definition}

\noindent We now reiterate the main result for computing bounds on delay.
\begin{theorem}[Maximum response time~\cite{chakraborty2003general}]
\label{thm:delay-bounds}
For a task dispatcher offering a minimum service curve $\beta(t)$ to a set of tasks $\mathcal{T}$ with an arrival curve $\alpha (t) $, the worst-case response time of a task is the maximum horizontal distance $hDev(\alpha,\beta)$ computed between $\alpha(t)$ and $\beta(t)$.  
\end{theorem}

\noindent To compute the response time of any priority, we use the service curve in Theorem~\ref{th:remainingService}.
\begin{theorem}[Minimum remaining service curve~\cite{wandeler2005real}]\label{th:remainingService}
For a preemptive fixed-priority dispatcher of computation capacity $\lambda$, and a set of tasks $\tau_i \in \mathcal{T}$ with priorities $p(i)$, the minimum service curve remaining to tasks of priority $p$ is the non-decreasing positive function
\begin{equation}
\beta^{SP}_{p} (t)= \nnup{\lambda\cdot t - \alpha_{>p}(t)}\ \text{, with~} \alpha_{>p}(t) 
    = \sum_{\tau_i \in \mathcal{T}, p(i) >p}\alpha_i(t).
\end{equation}
\end{theorem}

\noindent Similar to the NC considerations for sporadic flows and rate-latency servers~\cite{boyer_et_al:LIPIcs.ECRTS.2021.14}, 
if a task generates jobs of cost $C\in\mathbb{R}^+$ at a rate given by the period (or minimal inter-arrival time) $T\in\mathbb{R}^+$, it admits the linear arrival curve $\alpha_{r,b}: \mathbb{R}^+ \to \mathbb{R}^+$, $0 \mapsto 0$ and $t \mapsto rt + b$ if $t>0$, with $r = \frac{C}{T}$ and $b=r\cdot T$. When tasks are scheduled by a dispatcher in a certain slot, they are usually executed a constant rate $R$ (using one unit of computation for every unit of time) after some delay (latency) $L$ which is due to blocking by e.g. other higher-priority tasks. This matches a rate-latency service~\cite{boyer_et_al:LIPIcs.ECRTS.2021.14}, modelled by a function $\beta_{R,L}: t\mapsto R \cdot \nn{t-L}$. 

\begin{corollary}[Linear maximum response time]\label{cor:linearMaxResponse}
Consider a linear arrival curve $\alpha_{r,b}(t)$ and a rate-latency service curve $\beta_{R,L}(t)=  R\cdot \nn{t-L}$, then
$\hDev(\alpha_{r,b},\beta_{R, L})=L+\frac{b}{R}$.
\end{corollary}
 
\begin{prop}\label{prop:linearProperties}
Let $r,r',b,b',R,L \in \nnR$ be some parameters. Then, we have $\alpha_{r,b}(t)+\alpha_{r',b'}(t) = \alpha_{r+r',b+b'}(t) $.
If $r \leqslant R $, then,
   $\hDev(\alpha_{r,b},\beta_{R, L})=L+\frac{b}{R} $
    and
   $ \nnup{\beta_{R,L}(t)-\alpha_{r,b}(t)}     
    =
    \beta_{(R-r),\frac{RL+b}{R-r}}(t)$.\\
   \textnormal{The proofs can be found in~\cite[Prop.~3.7]{DNC-Book}.}
\end{prop}

\section{Related work}
\label{sec:related_work}
Sporadic events can be readily integrated alongside periodic real-time tasks in the schedulability analysis of both fixed- and dynamic-priority systems using, e.g., Deadline Monotonic (DM) or Earliest-Deadline First (EDF) schedulability tests (c.f~\cite{10.1145/321738.321743, 128746, Baruah2005AlgorithmsAC, 1336766, 4815215}). However, it may be helpful to have some form of isolation in the temporal domain between certain task types, e.g., between periodic/sporadic and aperiodic tasks. The temporal isolation is usually achieved via bandwidth servers, where a bandwidth server is defined as a periodic task with a budget and period that can handle one or more aperiodic events. The bandwidth server can then be scheduled alongside the periodic or sporadic tasks via, e.g., fixed-priority or dynamic-priority dispatchers. One of the main goals of bandwidth servers (beyond temporal isolation) is to minimize the response time of aperiodic tasks, but no guarantees such as deadlines can be given for aperiodic events. Examples of bandwidth servers for aperiodic ET task handling are the Polling Server~\cite{Lehoczky87, Sprunt89}, the Deferrable Server~\cite{strosnider95}, or the Sporadic Server in both the FP~\cite{Ghazalie95, Sprunt89} and EDF~\cite{spuri96} variants. 

In systems with a time-triggered scheduler, integrating sporadic event-based tasks is more challenging than in purely fixed- or dynamic-priority systems. In~\cite{894152, 840842}, the authors present a ``holistic'' schedulability analysis and design optimization approach for systems where tasks are event-triggered, but the communication backbone is based on the time-triggered bus protocol TTP. In~\cite{tpop2002, tpop2003}, a mixed time- and event-triggered application model similar to ours is considered for distributed embedded systems. The authors first present an analysis of ET task schedulability given a pre-defined TT schedule and then use a list-scheduling-based heuristic with limited backtracking to guide the generation of TT task schedules also to increase ET task schedulability. The approach in~\cite{tpop2002, tpop2003} favors the correctness of TT schedules over ET schedulability and is not guaranteed to find a feasible schedule for both TT and ET tasks. In that respect, it is more similar to a greedy method for generating TT schedules, checking the schedulability of ET in the process, albeit with an improved probability towards ET schedulability via different heuristics~\cite{tpop2003}. In contrast, we only accept solutions in which both TT and ET tasks are schedulable, starting from the schedulability of ET tasks to impose constraints on the TT schedule generation.

Hierarchical scheduling approaches such as~\cite{10.1145/1017753.1017772, Shin2003, Shin2008, 1019197} can be viewed as a more generalized form of a polling approach (c.f. Section~\ref{sec:polling_approach}), where on one level there is a fixed-priority scheduler for ET tasks, and on the underlying layer, a periodic resource abstraction (or periodic server) is used to decouple TT schedule generation from ET task schedulability analysis. Using the worst-case service pattern for the periodic resource abstraction, as is done in~\cite{10.1145/1017753.1017772}, has, on the one hand, the downside of the abstraction overhead (c.f.~\cite{Shin2008}) and, on the other hand, the server design problem~\cite{1212738} makes the problem difficult to solve, even for bandwidth-optimal approaches such as~\cite{4408298}. Moreover, for a mixed ET and TT system, there is the additional complexity of deciding how many polling tasks (i.e., resource abstractions/servers) to use and how to assign ET task subsets to the resources. The method proposed in~\cite{10.1145/1017753.1017772} is designed for constrained deadline ET tasks, whereas our affine envelope approximation considers arbitrary ET deadlines. Furthermore, we use a completely different approach based on affine envelope approximation that eliminates the need for solving the server design problem.

A third category of related work concerns the integration of TT and ET tasks at runtime, where the main goal is to minimize the response times of ET tasks. In~\cite{495205, 896010, isovic09}, a slot-shifting method is presented, which allows static TT schedule slots to be moved in order to execute sporadic and aperiodic tasks arriving dynamically at runtime. In~\cite{skalistis2019timely, skalistis2020dynamic} the schedule is safely adapted at runtime to allow for improved Quality-of-Service or the execution best-effort tasks. Moreover, in~\cite{isovic09}, the authors also include an offline analysis for sporadic ET tasks, which only looks at so-called critical slots in order to guarantee schedulability. However, the main assumption is that slot shifting is possible in the TT schedule, a property not implemented in many table-driven dispatchers. The method in~\cite{10.1016/j.sysarc.2019.101652}, while not directly modifying the TT schedule at runtime, has for a more flexible TT model using priority-based scheduler for ET tasks, similar to our system model, but allowing the TT schedule to be configured to an arbitrary priority in relation to the ET tasks. The work in~\cite{9183322} an SMT-based static schedule table synthesis is combined with an EDF-based online scheduler that handles sporadic ET tasks. While these approaches may work well for the average case, they do not mandate the schedulability of ET tasks and, additionally, impose a flexible execution model on the TT dispatcher.

\section{TT and ET integration using Polling Tasks}
\label{sec:polling_approach}
For sporadic tasks, a straightforward approach is to simulate or analyze the worst-case behavior of sporadic ET tasks for every (or selected) possible variation of the idle slots in the static scheduling table reserved for the polling task(s), similar to~\cite{tpop2002, tpop2003}. Hence, the schedule synthesis step has to check for any placement of the TT slots if the resulting idle slots used by the polling task to handle ET tasks are sufficient to fulfill the deadlines of ET tasks. A feedback loop could, in principle, offer some heuristic suggestions that may guide the search and, on average, speed up the algorithm. However, while heuristics like the ones provided in~\cite{tpop2002, tpop2003} may create correct schedules for both TT and ET tasks, there is no guarantee of this since they tend to favor TT schedule correctness over ET. If ET schedulability is crucial, a brute-force approach that will check every possible TT schedule for ET schedulability is intractable and will not scale for medium and large systems.
Hence, we introduce two polling-based methods inspired by prior work that can be used to guarantee sporadic ET tasks in systems with a static TT schedule table. 

\begin{figure}[!t]
		\centering
		\includegraphics[width=0.83\textwidth]{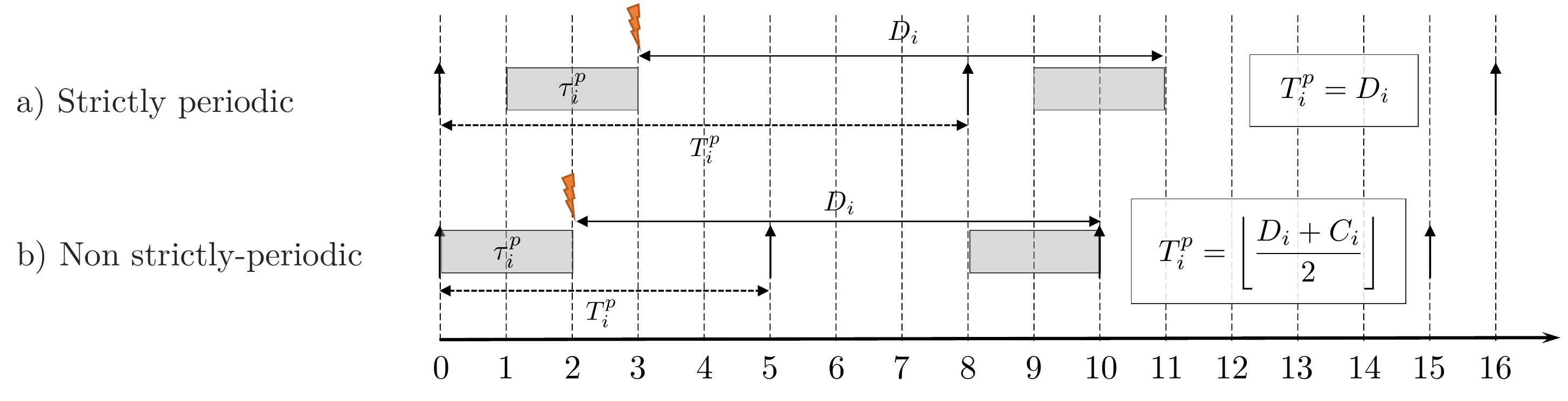}
		\caption{\label{fig:transformation} Worst-case event arrival wrt. slot placement.}
\end{figure}

A simple and computationally ``cheap'' method, which we call \emph{simple polling} \textbf{(SPoll)}, is to let each ET task $\tau_i \in \mathcal{T}^{ET}$ be handled by its own polling TT task $\tau^{p}_i$. If we know that the generated schedule table will be strictly periodic with respect to the placement of the polling task and the deadline is equal to the period, we can set the period of the polling task $\tau^{p}_i$ to $T^{p}_i = D_i$ and the computation time to $C^{p}_i = C_i$ (c.f. Figure~\ref{fig:transformation} (a)). In the case of non-strictly periodic slot placement (when using, e.g., EDF simulation to generate the schedule table~\cite{CraciunasETFA14}) and arbitrary ET deadlines, we have to use over-sampling to place the polling task in the static schedule table (c.f. Figure~\ref{fig:transformation}(b)). The over-sampling period $T_p$ is easily derived for e.g., out of the availability function in~\cite{10.1145/1017753.1017772} as $T^{p}_i = \lfloor\frac{D_i+C_i}{2}\rfloor$. For sporadic ET tasks with constrained deadlines, the polling task has a computation time $C^{p}_i = C_i$. For sporadic ET tasks with arbitrary deadlines, we need to consider how many previous job releases there can be within any polling period. Hence we have $C^{p}_i = \lceil \frac{T^{p}_i}{T_i} \rceil \cdot C_i$. This approach can be very pessimistic for tasks with a short deadline and long MIT/period or a long deadline and short period/MIT, leading to a reduced schedulability. Consider a similar example to the one from~\cite{Tindell97}, where a sporadic event with a computation time of $C_i = 2$ ms and a deadline of $D_i = 20$ ms needs to be handled. The event can occur at most once every $T_i = 100$ ms, thereby having a $2\%$ CPU utilization. However, because the exact arrival is not known, we would have to reserve a slot every $9$ ms (assuming a $2$ ms slot size) if we want to finish the execution within $20$ ms of any possible event arrival. This would consume $22,2\%$ of the CPU bandwidth. If the deadline is much larger than the period, e.g., $D_i = 100$ ms, $T_i = 10$ ms, and $C_i = 2$ ms, the period and the computation time of the polling task are $49$ and $10$, respectively, which results in a utilization of $20,4\%$. We see that, except in very simple systems, this approach results in a large over-utilization and will most likely not result in any feasible TT schedule creation. Both strictly periodic and non-strictly periodic approaches have the downside of reduced schedulability, either from oversampling or from the strictly periodic nature of ET slot placement.

A more precise approach, which we call \emph{advanced polling} \textbf{(AdvPoll)}, is a simplified version of the hierarchical scheduling paradigm~\cite{Shin2003, Shin2008, 1212738} with $2$ levels, a $2^{nd}$-level fixed-priority (FP) dispatcher for ET tasks, and a time-triggered dispatcher on the lowest level, similar to~\cite{10.1145/1017753.1017772}. Our reference method for the advanced polling is~\cite{10.1145/1017753.1017772} where the schedulability of a set of constrained deadline sporadic tasks is verified under a server with a given capacity and period. Similar to~\cite{10.1145/1017753.1017772}, we define a periodic resource abstraction (basically a budget and period) for each polling task such that the sporadic ET tasks are still schedulable if the polling task gets the desired budget in the given period. The offline schedule synthesis step for TT tasks can then readily include the polling task(s) when generating the schedule table as another periodic (set of) TT task(s), e.g., using exact methods or heuristics (c.f.~\cite{10.1023/A:1009804226473}). Naturally, there can be more than one polling task, each of them handling a disjoint subset of the ET tasks. To ease the notation, we assume for now that there is only one polling task $\tau_p$ handling the entire set $\mathcal{T}^{ET}$ of ET tasks for which $C_p$ and $T_p$ have to be determined. While in~\cite{10.1145/1017753.1017772} the polling task (periodic resource) is defined by a budget $C_p$ and a period $T_p$, a more general model called Explicit Deadline Periodic (EDP)~\cite{4408298, 10.1109/RTSS.2013.37, 1212738} can be used in which the server also has a deadline $D_p \leqslant T_p$. While this extension may increase the search space for possible TT schedules (and therefore schedulability), it will also result in a more complex server design problem (see below). 
We hence use the more simple model from~\cite{10.1145/1017753.1017772} to define the lower supply bound function $slbf(t)$ of a polling task $\tau^p$ in any time window of length $t \geqslant 0$. The exact expression of $slbf(t)$ can be found in~\cite{10.1145/1017753.1017772}, based on the characteristic function from~\cite{1212738}. To reduce complexity, the $slbf(t)$ is usually bound linearly from below by the so-called linear supply lower bound function $lslbf(t)$ (c.f.~\cite{1212738}) defined in~\cite{10.1145/1017753.1017772} using $a = \frac{C_p}{T_p}$ and $\Delta = 2 \cdot (T_p - C_p)$, as
\begin{equation}
\label{eqn:lslbf}
lslbf(t) = max\{0, (t-\Delta) \cdot a\}.
\end{equation}
Following the method in~\cite{10.1145/1017753.1017772}, we compute for each ET task $\tau_i \in \mathcal{T}^{ET}$ and for each instant $t$ the maximum load of task $\tau_i$ and all higher and equal priority tasks (maximum load of level-i) $H_i(t)$. We can use the classical definition of the maxium load of level-i from~\cite{63567} for constrained deadline tasks, namely 
\begin{equation}
\label{eqn:loadlevel}
H_i (t) = \sum_{\forall \tau_j \in \mathcal{T}^{ET}, p_j \ge p_i} \left \lceil \frac{t}{T_j} \right \rceil \cdot C_j.
\end{equation}

The schedulability condition for an ET task $\tau_i \in \mathcal{T}^{ET}$ is defined in Lemma 1 of~\cite{10.1145/1017753.1017772}. The worst-case response time $R_i$ for task $\tau_i \in \mathcal{T}^{ET}$ can be calculated by determining the earliest time instant in which the maximum load of level-i $H_i(t)$ intersects the linear supply bound function $lslbf(t)$ of the polling task from Eq.(\ref{eqn:lslbf}), as follows~\cite{10.1145/1017753.1017772}:
\begin{equation}
R_i = \text{earliest } t: t = \Delta + H_i(t)/\alpha.
\end{equation}

Using this method, the parameters of the polling task need to be found (which is, in essence, the non-trivial server design problem~\cite{Shin2003, 10.1145/1017753.1017772}) and then consider the polling task as a regular TT task alongside the other TT tasks in the system when creating the schedule table. The schedule generation step can be done relatively efficiently by simulating EDF/LLF scheduling until the hyperperiod of the TT tasks (e.g.~\cite{McLeanETFA20, CraciunasETFA14}), especially for harmonic TT task periods where the hyperperiods are (relatively) small~\cite{6932603} or when period re-dimensioning is possible to reduce the hyperperiod~\cite{6059178, 7176034}. However, the main drawback of this approach is that we have first to decide how many polling tasks to use, how to split ET tasks between them, and then for each polling task, find the computation time, the period, and the deadline. While there are specific optimizations that can be employed (e.g., using external points~\cite{10.1145/1017753.1017772} or the methodology from~\cite{1212738}), the approach can be computationally intensive for large systems as the assignment of ET tasks to polling tasks is in itself a combinatorial problem. 

In classical hierarchical scheduling, the aim is to find the resource abstraction with the least impact on other components, i.e., the best $(C_p, T_p)$ where the utilization is just large enough to respect the ET deadlines. The search for the best $(C_p, T_p)$ may be complex since we have to iterate not only through $T_p$, but for every $T_p$, we need to find $C_p$ (and potentially $D_p$). We can use a simplification here in order to get rid of the binary search for $C_p$ for every $T_p$ since we can use the maximum $C_p$ that does not lead to an overutilization, i.e., $C_p = \lfloor (1-U^{TT})\cdot T_p\rfloor$. However, we note that this optimization only applies if TT tasks have implicit deadlines, ET tasks have constrained deadlines, and if there is only one polling task with $D_p = T_p$, whereas our method also works for constrained-deadline TT tasks and arbitrary-deadline ET tasks. 

\section{TT and ET integration using affine envelope approximations}
\label{sec:tt_envelope}
The main idea of our method is to derive a constraint on the TT schedule that will guarantee ET task schedulability and then use this constraint to build a correct TT schedule. First, the constraint is expressed as a maximal affine envelope for the TT tasks, computed such that as long as a TT schedule respects this envelope (expressed as token-bucket arrival curve), all ET tasks respect their deadlines (Section~\ref{sec:SPTT:TT-max-env}). The second step consists in building a TT schedule generation algorithm that enforces the envelope while maintaining TT task schedulability (Sections~\ref{sec:blc} and~\ref{sec:blllf}).

Since we know the TT task set, the utilization rate of TT tasks $U^{TT}$ is known, but the burst $b^{TT}$ of the linear arrival curve $\alpha^{TT}(t)$ is unknown and depends on the future schedule. Furthermore, as TT has a higher priority than ET, this burst $b^{TT}$ impacts the ET response times. Hence, the goal of our method is first to identify the maximum burst $ b^{TT}_{max}$ such as the ET tasks fulfill their deadlines, and then to compute a TT schedule such as an arrival curve of the scheduled TT tasks is $\alpha^{TT}(t)=U^{TT}\cdot t + b^{TT}_{max}$. To do so, we first evaluate the impact of the TT tasks on the ET tasks and compute $b^{TT}_{max}$ in Section~\ref{sec:SPTT:TT-max-env}. Then, in Sections~\ref{sec:blc} and~\ref{sec:blllf} we present a scheduler capable of enforcing $\alpha^{TT}(t)=U^{TT}\cdot t + b^{TT}_{max}$.

\subsection{Computing a maximal affine envelope for TT tasks}\label{sec:SPTT:TT-max-env}

To compute the maximum TT burst such as the ET task deadlines are fulfilled, we first calculate the worst-case response time (i.e., delay) depending on the TT burst and TT utilization rate in Theorem~\ref{th:ET-bound:TT-burst}, then we deduce the maximum admissible TT burst in Theorem~\ref{th:maxTTBurst}.

\noindent First, we can bound the TT burst as defined in Theorem~\ref{th:pessimisticTTBurstBound}.
\begin{theorem}[Worst-case burst for TT tasks]\label{th:pessimisticTTBurstBound}
    The function $\alpha_{U^{TT},C^{TT}}$ is an arrival curve for the set of TT tasks $\mathcal{T}^{TT}$.
\end{theorem}
 
\begin{proof}
The functions $\alpha_{U_i,C_i}$ are arrival curves for the TT tasks $\tau_i$. So an arrival curve for the set of TT tasks $\tau$ is $\alpha_\tau=\sum_{\tau_i}\alpha_{U_i,C_i}=\alpha_{U^{TT},C^{TT}}$ according to Proposition~\ref{prop:linearProperties}. 
\end{proof}
This is a (pessimistic) burst that will be refined further. A similar result has been presented in~\cite[Thm.~2]{NC-AVB-Boyer}, but a direct proof is given here for completeness.

\begin{theorem}[Response time of ET tasks]
\label{th:ET-bound:TT-burst}
Let $\alpha_{U^{TT},b^{TT}}$ be the arrival curve of the aggregated scheduled TT tasks, and $\alpha_{U_i,C_i}$ the linear arrival curve of each ET task $\tau_i\in \mathcal{T}^{ET}$. The maximum response time of a ET task $\tau_i$ of priority $p(i)$ is:
\begin{equation}
    \hDev\left( 
    \alpha_{p}^{ET}, \beta^{SP}_{p}
    \right)= \frac{b^{TT}+C^{ET}_{\geqslant p}}{\lambda-U_{>p}}
\end{equation}
\end{theorem}

\noindent The hyperperiod does not appear in the expression since we use an overapproximation through affine functions in which we only require the individual task periods to compute this worst-case bound (e.g., \cite{Bini-SP-Continuous} also have a bound independent of the hyperperiod).

\begin{proof}

Let $(\alpha_i,\ldots,\alpha_{n^{ET}})$ the linear arrival curves of the ET tasks, with $n^{ET}=\card{\mathcal{T}^{ET}}$. Let the arrival curves of the aggregated tasks of priority p and of priorities >p be respectively:
  \begin{gather}
    \begin{aligned}
    \alpha_{p}^{ET} (t) & 
    \isdef \sum_{\tau_i \in \mathcal{T}^{ET}, p(i)=p} \alpha_i(t) 
    ~~~~~~~
    & 
    \alpha^{ET}_{>p}(t)  & 
    \isdef \sum_{\tau_i \in \mathcal{T}^{ET}, p(i) >p}\alpha_i(t) , ~~~~ 
    \end{aligned}
  \end{gather}

\noindent The online scheduler uses preemptive fixed-priority scheduling, so $\beta^{SP}_p(t)$ is a service curve for the task of priority $p$ (cf. Theorem~\ref{th:remainingService}) and $\hDev\left( \alpha_{p}, \beta^{SP}_{p} \right)$ is an upper bound on the delay of each task in $\mathcal{T}^{ET}_p$ (cf. Theorem~\ref{thm:delay-bounds}). We also consider the linear arrival curves for ET tasks, 
\begin{gather}
    \begin{aligned}
 \alpha^{ET}_{p} 
  & =  \sum_{\tau_i \in \mathcal{T}^{ET}, p(i) =p} \alpha_{U_i,C_i}
  ~~~~~~~
    & 
  \alpha^{ET}_{>p} & =  \sum_{\tau_i \in \mathcal{T}^{ET}, p(i) >p} \alpha_{U_i,C_i}
  \end{aligned}
    \end{gather}
 From the definitions in Section~\ref{sec:model}  and  Proposition~\ref{prop:linearProperties}, we have:
 \begin{gather}
    \begin{aligned}
    \alpha^{ET}_{p}  & = \alpha_{U^{ET}_{= p},C^{ET}_{=p}}
    ~~~~~~~~~~
    & 
    ~~~~~~~~~~
    & 
    \alpha^{ET}_{>p}  & = \alpha_{U^{ET}_{> p},C^{ET}_{>p}} 
 \end{aligned}
  \end{gather}
As TT priority is higher than ET priorities, we have $U^{TT}+U^{ET}_{>p}=U_{>p}$, which gives:
\begin{equation}
\alpha^{ET}_{>p} + \alpha_{U^{TT},b^{TT}}  =
\alpha_{U_{>p},b^{TT}+C^{ET}_{>p}}
\end{equation}
Since all TT tasks have a higher priority than  ET priority $p$, the sum of arrival curves with priority higher than $p$ is $\alpha^{ET}_{>p}+\alpha_{U^{TT},b^{TT}}$. Hence, the residual service can be expressed, using Theorem~\ref{th:remainingService} and Proposition~\ref{prop:linearProperties}:
  \begin{equation}\label{eq:betasp}
      \beta^{SP}_p(t)= \nnup{\lambda\cdot t - \alpha^{ET}_{>p}(t) -\alpha_{U^{TT},b^{TT}}(t)}
    =  \beta_{(\lambda-U_{>p}),\frac{b^{TT}+C^{ET}_{>p}}{\lambda-U_{> p}}}(t)
  \end{equation}
Finally, with Proposition~\ref{prop:linearProperties}, we have:
\begin{equation}
\hDev(\alpha^{ET}_p,\beta_{(\lambda-U_{>p}),\frac{b^{TT}+C^{ET}_{>p}}{\lambda-U_{>p}}}) = \frac{b^{TT}+C^{ET}_{>p}}{\lambda-U_{>p}} + \frac{C^{ET}_{=p}}{\lambda-U_{>p}} = \frac{b^{TT}+C^{ET}_{\geqslant p}}{\lambda-U_{>p}}.\label{eq:hDev}
\end{equation}
\end{proof}

\noindent We now define the maximum admissible TT burst such as ET tasks fulfill their deadlines in Theorem~\ref{th:maxTTBurst}.

\begin{theorem}[Maximal admissible TT burst]\label{th:maxTTBurst}
Let $\alpha_{U^{TT},b^{TT}}$ be the arrival curve of the aggregated scheduled TT tasks, and $\alpha_{U_i,C_i}$ the linear arrival curve of each ET task $\tau_i\in \mathcal{T}^{ET}$ with a priority $p(i)$.
The maximum value of $b^{TT}$ fulfilling the deadlines of all ET tasks is:
\begin{equation}
b^{TT}_{max}= \min\bigg(\min_{p(i)}\Big( (\lambda-U_{>p(i)})\cdot(\min_{\forall \tau_j ,p(i)=p(j)}D_j)-C^{ET}_{\geqslant p(i)}\Big), C^{TT} \bigg)
\end{equation}

\end{theorem}

\begin{proof}
From Theorem~\ref{th:ET-bound:TT-burst}, we know that if $\forall\tau_j$, $$\frac{b^{TT}+C^{ET}_{\geqslant p(i)}}{\lambda-U_{>p(i)}} \leqslant \min_{p(i)=p(j)} D_j,$$ all ET tasks of priority $p(i)$ respect their deadlines. So $$b^{TT}\leqslant \min_{p(i)}\Big((\lambda-U_{>p(i)}) \cdot(\min_{p(i)=p(j)} D_j)-C^{ET}_{\geqslant p(i)}\Big)$$ and we know from Theorem~\ref{th:pessimisticTTBurstBound} that $b^{TT}\leqslant  C^{TT}$.
\end{proof}

For methods using the lower supply bound function, i.e. minimum service curve, we note that $lslbf(t)$ (c.f. Eq.~(\ref{eqn:lslbf})) uses the same linear approximation as $\beta^{SP}_p(t)$ (c.f. Eq.~(\ref{eq:betasp})). However, they are built under different hypotheses. The $lslbf(t)$ is computed considering only the slot duration and period that will be assigned to the polling tasks $\tau_p$, whereas  $\beta^{SP}_p(t)$ considers the impact of higher priority tasks on the current set of ET tasks of priority $p$. In $lslbf(t)$ the unknown values are $T_p$ and $C_p$, if we do not take the simplifications that lead to more pessimism explained in Section~\ref{sec:polling_approach}. In $\beta_{SP}(t)$ from our approach the unknown value is $b^{TT}$. However, after the unknown variables are computed, in both cases, the functions will lead to worst-case delays lower than the deadlines. As for the maximum requested computation, we described the ready tasks using a linear approximation $\alpha_{r,b}(t)$, instead of the more precise staircase function $H_i(t)$ defined in Eq.~(\ref{eqn:loadlevel}).

\begin{figure}[!t]
    \centering
    \includegraphics[width=0.99\linewidth]{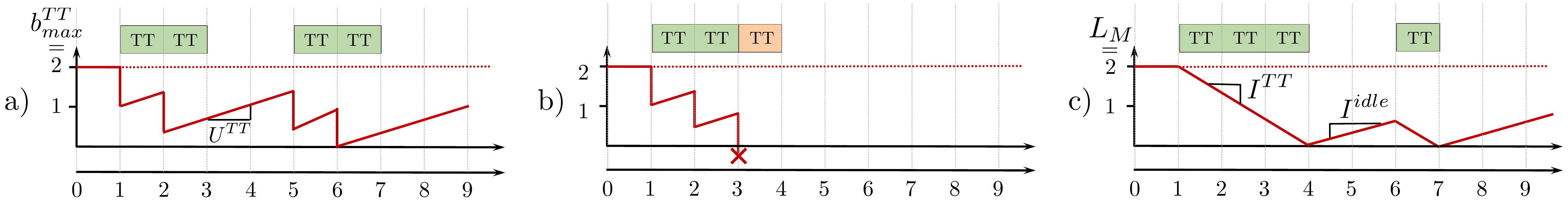}
    \caption{Checking a TT schedule under the TB and BLC constraints}
    \label{fig:emulatinTB}
\end{figure}

\subsection{Burst Limiting Constraint (BLC)}
\label{sec:blc}
Traditionally, a token bucket \textbf{(TB)} (or a leaky bucket) shaper would be used to check the computed TT envelope. Under a TB, the budget for a slot (i.e., 1) is paid at the slot allocation time $t$, while the budget continuously increases with a rate $I^{idle}$ as illustrated in Figure~\ref{fig:emulatinTB}~(a). This means that the budget is a non-continuous function. However, we are trying to check a continuous function $\alpha^{TT}(t)=U^{TT}\cdot t +b^{TT}_{max}$. Hence, between the moment the budget is paid and the end of the slot, the arrival curve function we are trying to conform to has increased by $\alpha^{TT}(1)=U^{TT}+b^{TT}_{max}$. In Figure~\ref{fig:emulatinTB}~(b), in interval $[1,4]$, with a maximum burst $b_{max}^{TT}=2$ and a replenishment rate of $U^{TT}=1/3$, the maximum allowed cumulative processing done for TT is $3$. Hence, the third TT task at $t=3$ does conform to $\alpha^{TT}(t)$, but the TB is too pessimistic due to paying the full budget at the start of the slot.
Due to this, the TB is only optimal for infinitesimally small resource demand granularity~\cite{5981848}. 

To improve this issue, we introduce a so-called \emph{Burst Limiting Constraint} \textbf{(BLC)}, inspired by the Burst Limiting Shaper (BLS)~\cite{Gotz2012, 7092358} and the Credit Based Shaper (CBS)\cite{AVBcbs}. Instead of paying the budget at the start of the slot, we check that the budget at the end of the slot conforms to the maximum arrival curve, and we pay the budget continuously during the slot at a rate $I^{TT}=1-U^{TT}$. Hence, at the end of a TT slot, the budget variation is the same as with a TB, but without the discontinuity of the budget. We detail the BLC in Definition~\ref{def:blc}, and we will show in Theorem~\ref{th:maximumArrivalBLC} that the proposed BLC offers a maximum service curve that can be easily parameterized to fit $\alpha^{TT}(t)$.

\begin{definition}[Burst Limiting Constraint]\label{def:blc}
Given a TT schedule $\sigma$, we define a Burst Limiting Constraint (BLC) such that a slot reserved for TT in $\sigma$ is invalid if the budget is strictly smaller than $0$ at the end of the slot, with the budget defined as follows:
\begin{itemize}
    \item the budget $bdg_{\sigma}(t)$  is a continuous piecewise linear function of the time t $\in \mathbb{R}^+ \mapsto \mathbb{R}$,  
    \item when a slot is reserved for TT in $\sigma$, the budget decreases at a rate $bdg'_{\sigma}(t)=-I^{TT}<0$,
    \item when a slot is idle in $\sigma$ (i.e., not assigned to TT), the budget increases at a rate $bdg'_{\sigma}(t)=I^{idle}>0$ while the budget is strictly smaller than a maximum value $L_M$, or else it remains constant at $L_M$, i.e. $bdg'_{\sigma}(t)=0$,
    \item the sum of $I^{TT}$ and $I^{idle}$ is the processing capacity: $\lambda$,
    \item at time $0$, the budget is $L_M$, i.e. $bdg_\sigma(0)=L_M$.
\end{itemize}
\end{definition}
The unit of the budget is the computation unit, the unit of $I^{TT}$ and $I^{idle}$ is computation units per time unit. The BLC budget variations are illustrated in Figure~\ref{fig:emulatinTB}~(c). We can see that in the interval $[1,4]$, we are able to assign $3$ slots, which is the maximum amount allowed by $\alpha^{TT}(t)=1/3\cdot t+2$ for an interval of duration $3$. In this respect, the BLC does better than the TB in Figure~\ref{fig:emulatinTB}~(b). However, in the interval $[0,6]$, with a first slot idle, there can be only $3$ TT slots due to the saturation of the budget between in $[0,1]$. So, while the BLC itself is not optimal either, its performance is better than the TB, and this difference can significantly impact schedulability. As visible in Figure~\ref{fig:emulatinEDF-under-TB}~(b) vs.~\ref{fig:emulatinEDF-under-TB}~(c), for a maximum burst $b=2$ under TB, the schedule is invalid, but when transforming the TB to a BLC with the same burst ($L_M$), the schedule becomes valid.

While the BLC resembles the CBS and BLS by its use of a budget/credit, Definition~\ref{def:blc} shows that it is quite different from either of them. With the BLC and contrary to the credit of the CBS \cite{imtiaz2009performance}, the budget is continuous, and we set the budget upper and lower bounds. Moreover, while the BLS has these $3$ properties, with the BLC, there is no priority inversion at a defined level $L_R$, and there can be no saturation of the budget at 0 \cite{finzi2018incorporating}.
 
We now present the maximum service curve offered to TT tasks by the BLC.

\begin{theorem}[BLC maximum service curve]\label{th:maximumArrivalBLC}
The maximum service curve of a set of scheduled TT tasks validated by the Burst Limiting Constraint (BLC) defined in Def.~\ref{def:blc} is
$\gamma_{blc}^{TT}(t) = I^{idle} \cdot t + L_M.$
\end{theorem}

\noindent \textbf{Proof.}
The proof is based on the proofs detailed in~\cite{finzi2018incorporating, finzi2020worst} for the Burst Limiting Shaper (BLS)~\cite{Gotz2012, 7092358} in TSN networks. We denote $\mathcal{C}^{TT} (t)$ the computation capacity function offered to TT tasks, and $\Delta \mathcal{C}^{TT}(t,\delta)= \mathcal{C}^{TT}(t+\delta)-\mathcal{C}^{TT}(t)$ its variation during an interval $\delta\geqslant 0$, and $\lambda$ the total processing capacity (in computation units per time unit). Hence, $\frac{\Delta \mathcal{C}^{TT}(t, \delta)}{\lambda}$ represent the executing time of the tasks TT during any interval $\delta$. According to Definition~\ref{def:service-curves}, we search $\gamma_{blc}^{TT}(\delta)$ such as $\Delta \mathcal{C}^{TT}(t, \delta) \leqslant \gamma_{blc}^{TT}(\delta), \forall t\geqslant0 $. 

We consider a known TT schedule $\sigma$. If $\sigma$ fulfills the BLC, then we know that $bdg_{\sigma}(t)\geqslant 0, \forall t\geqslant 0$ and that the budget cannot saturate at 0: if the budget is 0, the next slot will be idle to fulfill the BLC and so the budget will increase. Therefore, there are three possible variations of the budget: 1) the budget increases when a slot is idle, and the budget is strictly smaller than $L_M$; 2) the budget decreases when a slot is assigned to TT in $\sigma$; 3) the budget saturates at $L_M$ when a slot is idle, and the budget is already at $L_M$. Hence, we denote $\Delta \mathcal{C}_{L_M, sat}(t, \delta)$ the number of computation units where the budget is saturated at $L_M$.

We present here a lemma linked to the budget saturation and necessary for the maximum service curve proof. In Lemma \ref{minmaxsum}, we show how to bound the sum of the budget consumed and the budget gained, depending on the budget saturation. 

\begin{lemma}[Continuous budget bounds]\label{minmaxsum}
$\forall$ set of assigned TT tasks fulfilling a BLC, $\forall t\geqslant 0,\delta\geqslant 0$, the variation of the computation capacity $\Delta C^{TT}(t, \delta)$ is bounded by:
$$-L_M  \leqslant  -	
			\Delta C^{TT}(t, \delta)+(\delta - \frac{\Delta C_{L_M,sat}(t, \delta)}{\lambda})\cdot I^{idle}
		  \leqslant  L_M$$
\end{lemma}
\begin{proof}
	
In an interval $t, t+\delta$, for any set of assigned TT tasks fulfilling the BLC, the accurate consumed budget is the duration corresponding to the slots $\frac{\Delta \mathcal{C}^{TT}(t, \delta)}{\lambda}$ multiplied by the  signed TT slope:
$$budget_{consumed}= \frac{\Delta \mathcal{C}^{TT}(t, \delta)}{\lambda} \cdot (-I^{TT}).$$

Conversely, the gained budget is the remaining time $\delta - \frac{\Delta \mathcal{C}^{TT}(t, \delta)}{\lambda}$ minus the saturation time $\frac{\Delta \mathcal{C}_{L_M,sat}(t, \delta)}{\lambda}$, multiplied by the idle slope:	 
$$budget_{gained}=\left( \delta - \frac{\Delta \mathcal{C}^{TT}(t, \delta)+\Delta \mathcal{C}_{L_M,sat}(t, \delta)}{\lambda}\right) \cdot I^{idle}.$$

Thus $\forall \delta\in \mathds R^+$, using the fact that $I^{TT}+I^{idle}=\lambda$, the sum of the gained and consumed budget, expressed as $budget_{consumed}+budget_{gained}$, is: $$-\Delta \mathcal{C}^{TT}(t, \delta)+\Big(\delta - \frac{\Delta \mathcal{C}_{L_M,sat}(t, \delta)}{\lambda}\Big)\cdot I^{idle}.$$
\noindent Since the budget is a continuous function with lower and upper bounds 0 and $L_M$, respectively, the sum of the consumed and gained budget is always bounded by $-L_M$ and $+L_M$:
$$-L_M  \leqslant  -	
			\Delta C^{TT}(t, \delta)+\Big(\delta - \dfrac{\Delta C_{L_M,sat}(t, \delta)}{\lambda}\Big)\cdot I^{idle}
		  \leqslant  L_M
		$$
\end{proof}
Returning to the proof of Theorem~\ref{th:maximumArrivalBLC}, we know from Lemma~\ref{minmaxsum} that $-L_M \leqslant -\Delta \mathcal{C}^{TT}(t, \delta)+(\delta - \frac{\Delta \mathcal{C}_{L_M,sat}(t, \delta)}{\lambda})\cdot I^{idle}$. Thus, $$\Delta \mathcal{C}^{TT}(t, \delta)\leqslant L_M+\Big(\delta - \dfrac{\Delta \mathcal{C}_{L_M,sat}(t, \delta)}{\lambda}\Big)\cdot I^{idle}.$$ We know by definition that: $ \Delta \mathcal{C}_{L_M, sat}(t, \delta)\geqslant 0$.
Hence, we obtain
$$\Delta \mathcal{C}^{TT}(t, \delta)  \leqslant I^{idle}\cdot\delta + L_M= \gamma_{blc}^{TT}(\delta).$$

\subsection{Burst Limiting Least Laxity First (B3LF)}\label{sec:blllf}

While checking that an existing TT schedule adheres to the BLC (or TB) is easy\footnote{Given a TT schedule, it is easy to check that it respects a given token-bucket constraint, as illustrated in Figure~\ref{fig:emulatinTB} (a) and (b). This can be done in linear time with regards to the schedule length, and if the token-bucket shape of a schedule is known, it can be updated in case of update of the schedule without a complete re-computation.}, the more interesting (and useful) question is how to create TT schedules that respect the BLC constraint. A first idea would be to emulate well-known mechanisms such as Earliest-Deadline-First (EDF)~\cite{10.1145/321738.321743} or Least-Laxity-First (LLF)~\cite{26230a7429bc45d6adf9ae7fcc590fe6} while keeping the schedule under the BLC, as illustrated in Figure~\ref{fig:emulatinEDF-under-TB}~(a). However, it is easy to show a counterexample (c.f. Figure~\ref{fig:emulatinEDF-under-TB}~(a) vs. Figure~\ref{fig:emulatinEDF-under-TB}~(b)) proving that it is not an optimal result and it may lead to deadline misses. The problem with using EDF/LLF is that it can reach the maximum allowed burst by scheduling a task immediately if, e.g., it is the only one in the ready queue, even though it has enough slack and could be executed later. Thus, at the next time instant, there is no more available budget, and we cannot schedule any 0-slack TT task that has been released, leading to a deadline miss. This problem will persist under any work-conserving algorithm since sometimes it may be necessary to insert idle times to have the full burst at a later time when it may be needed.

\begin{figure}[!t]
    \centering
    \includegraphics[width=0.99\linewidth]{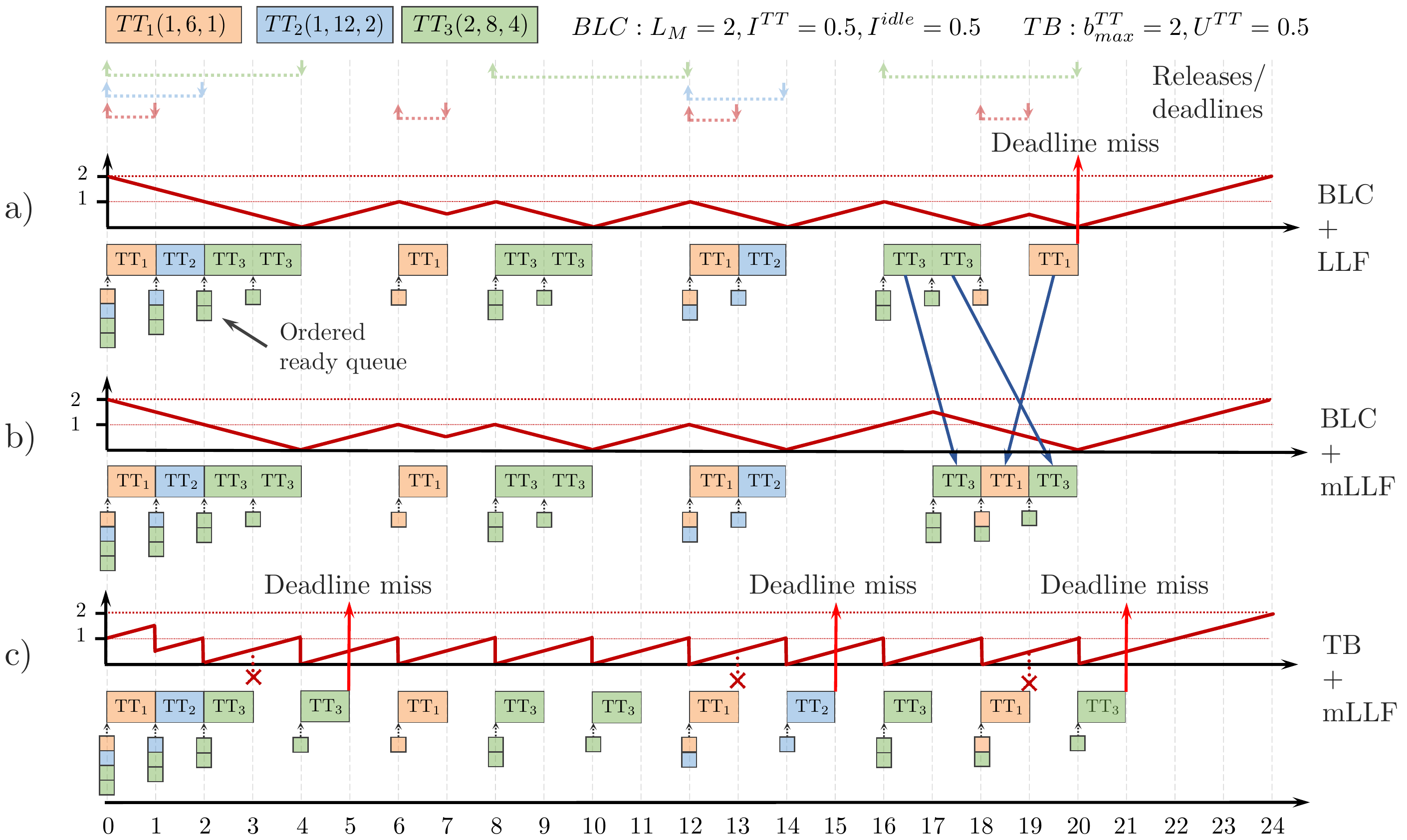}
    \caption{(a) LLF under a BLC constraint vs. (b) mLLF under a BLC constraint vs. (c) mLLF under a TB constraint.}
    \label{fig:emulatinEDF-under-TB}
\end{figure}

To solve this problem, we use the BLC to enforce the ET constraints through the maximum burst of TT, and we combine it with a modified LLF \textbf{(mLLF)} algorithm as described in Section~\ref{sec:blllf}. Together, the BLC and the mLLF algorithm result in a scheduler that enforces both TT and ET tasks deadlines. We call our method the \emph{Burst Limiting Least Laxity First} scheduler \textbf{(B3LF)}. As we will detail in Section~\ref{sec:blllf}, the B3LF algorithm respects the proposed BLC by construction. The main idea of the following analysis is that mLLF itself does not negatively constrain the TT tasks, so by modeling the BLC, we are able to model the TT constraint enforced by the whole B3LF. Hence, to model the B3LF in RTC, we separate it into its two components: the BLC and the mLLF, as illustrated in Figure~\ref{fig:bledf}. 
The B3LF is executed offline to create a static schedule table such that the TT arrival curve at runtime is $\alpha_{sp}^{TT}\leqslant \alpha^{TT}(t)=U^{TT}\cdot t+b^{TT}_{max}$, to enforce the ET deadlines (Theorem~\ref{th:ET-bound:TT-burst}).

In our model, the schedule is the output of the mLLF which itself depends on the BLC. Thus $\alpha_{sp}^{TT}$ is limited by the maximum service curve of the B3LF $\gamma^{TT}_{b3lf}(t)$, which is the minimum of the maximum service curves of the BLC $\gamma^{TT}_{blc}(t)$ and mLLF $\gamma^{TT}_{mllf}(t)$ : \begin{equation}
\alpha_{sp}^{TT}(t) \leqslant \gamma^{TT}_{b3lf}(t) = \min(\gamma^{TT}_{blc}(t),\gamma^{TT}_{mllf}(t))
\end{equation}
However, as will be shown in Theorem~\ref{th:maximumArrivalLLF} in Section~\ref{sec:bl3f}, the maximum service offered by the mLLF to TT tasks is only limited by the CPU processing capacity $\lambda$. So the TT input arrival curve in SP can only be constrained under $\lambda\cdot t$ by the BLC:
\begin{equation}\label{eq:alphaTTgammablc}
\alpha_{sp}^{TT}(t)  \leqslant \gamma^{TT}_{blc}(t)=I^{idle}\cdot t + L_M
\end{equation}
Finally, from the ET tasks, we have computed the affine TT envelope $\alpha^{TT}(t)$ to enforce the ET deadlines. As presented previously, we set the BLC parameters $I^{idle}=U^{TT}$, $I^{TT}=\lambda - U^{TT}$ and $L_M=b^{TT}_{max}$ and so we obtain a schedule with $\alpha_{sp}^{TT}(t)\leqslant \alpha^{TT}(t)=U^{TT}\cdot t+b^{TT}_{max}$, as illustrated in Figure~\ref{fig:bledf}.
 
\begin{figure}[!t]
    \centering
    \includegraphics[width=0.95\linewidth]{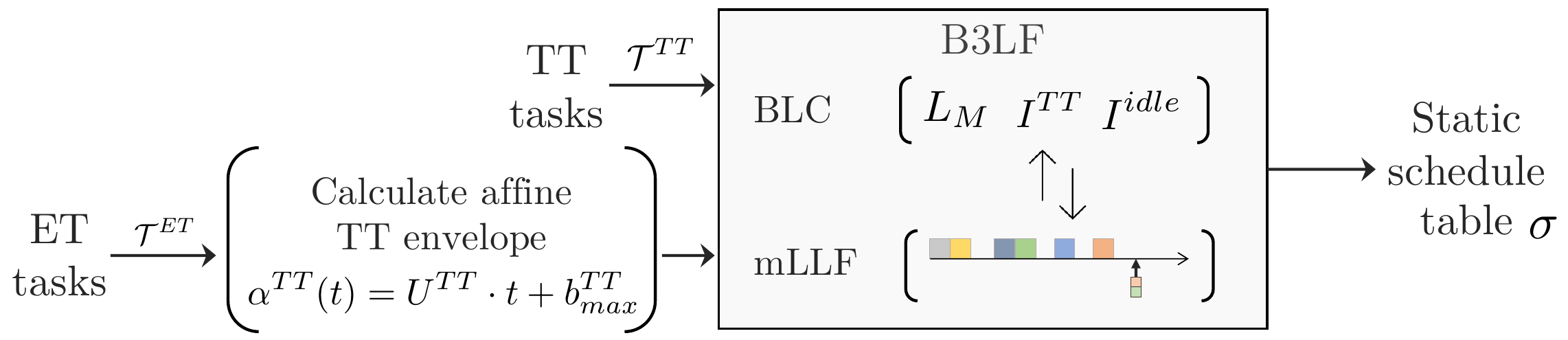}
    \caption{Burst Limiting Least Laxity First Scheduler (B3LF) and BLC parameterization}
    \label{fig:bledf}
\end{figure}

\subsubsection{B3LF Algorithm}\label{sec:bl3f}
The BLC ensures that the ET constraints are met by enforcing the TT slot allocation according to Definition~\ref{def:blc} with the budget replenishment and consumption rates given by $I^{idle}$ and $I^{TT}$, respectively, within the budget bounds 0 and $L_M$. However, we also need to ensure that the TT tasks are schedulable. Employing a standard LLF (or EDF) without BLC will only result in the schedulability of the TT tasks as shown below (Theorem~\ref{th:maximumArrivalLLF}).

\begin{theorem}[LLF and mLLF maximum service curve for TT tasks]\label{th:maximumArrivalLLF}
The maximum service curve of a set of TT tasks scheduled using the Least Laxity First, with (mLLF) or without (LLF) our proposed modification to add Idle tasks, for a processing capacity of $\lambda$, is:
\begin{equation}
    \gamma_{llf}^{TT}(t) = \gamma_{mllf}^{TT}(t) = \lambda \cdot t
\end{equation}
\end{theorem}
\begin{proof}
If utilization $U^{TT}=\lambda$, then LLF/mLLF assigns all the slots to TT, and the full processing capacity is used by TT, so, according to, Definition~\ref{def:service-curves}, we have
$$\mathcal{C}(t)-\mathcal{C}(s)\leqslant \lambda \cdot(t-s)=\gamma_{llf}^{TT}(t-s)=\gamma_{mllf}^{TT}(t-s), \forall s\leqslant t.$$
\end{proof}

Hence we introduce a modified LLF scheduler (mLLF) that constructs the static schedule table under the given BLC such that the TT task deadlines are met. LLF (Least Laxity First)~\cite{26230a7429bc45d6adf9ae7fcc590fe6} assigns dynamic priorities according to the current task laxity, i.e., tasks with smaller laxity are assigned a higher priority. We use LLF (instead of, e.g., EDF) because it better suits our need to track the BLC budget consumption and replenishment at any instant on the discrete timeline since LLF is a job-level dynamic priority algorithm (as opposed to EDF, which keeps priorities fixed at job-level). However, we note that a modified job-level dynamic EDF algorithm may also be a viable approach. Moreover, some practical runtime issues associated with LLF (e.g., the complex implementation or runtime calculation of laxity values) are not a concern here since we only use LLF offline to generate static schedules. Additionally, the high preemption overhead associated with ``thrashing'' when multiple tasks have the same laxity can be mitigated either by post-processing the resulting schedule table or by directly using modifications like ELLF~\cite{777467}. ELLF aims to resolve situations in which thrashing between tasks typically occurs by executing them consecutively via excluding all but one task that would thrash from consideration. ELLF can be used straightforwardly, except that the special IDLE task is always selected based on the least laxity value and is never excluded. Post-processing would analyze the resulting schedule, identify thrashing situations, and exchange execution slots of thrashing tasks to minimize preemptions.

Our mLLF scheduler (Algorithm~\ref{alg:schedulerBudget}) works very similar to the standard algorithm in that at each point in time $t$, we compute the slack (or laxity) of a task $\tau_i \in \mathcal{T}^{TT}$ as $L_i(t) = D_i(t) - C_i(t),$ where $D_i(t)$ represents the duration from time $t$ to the next deadline of the task, and $C_i(t)$ represents the remaining computation time at time $t$. In addition to the TT tasks that are considered by our mLLF, we introduce a special \emph{IDLE} task, denoted $\tau_{IDLE}$ that is responsible for introducing idle slots into the schedule $\sigma$. The computation time, period, and deadline of $\tau_{IDLE}$ are irrelevant since it will always be active and, when selected for execution, will introduce an idle slot into the schedule. Let us denote the current budget of the BLC with $bdg_{\sigma}(t)$. The main aspect of the idle task is its laxity computed as 
\begin{equation}
 L_{IDLE}(t) =
  \begin{cases}
    \left \lfloor \dfrac{bdg_{\sigma}(t)}{I^{TT}} \right \rfloor I^{TT}      & \quad \text{ if } bdg_{\sigma}(t) < L_M - I^{idle}\\
    \infty  & \quad \text{ otherwise}
  \end{cases}
\end{equation}
The laxity of $\tau_{IDLE}$ at some time $t$ is the amount of time until an idle slot must be scheduled because the budget will reach $0$ when scheduling only $TT$ tasks. Hence, the closer the budget is to $0$, the higher priority $\tau_{IDLE}$ becomes, but the scheduler still allows a TT task to be scheduled if necessary. In this way, we make sure that we stay within the budget constraints of the BLC but also steer the mLLF to prefer scheduling idle slots whenever the laxity of TT tasks permits it. It is interesting to note that this customization does not change the maximum service offered to TT tasks, i.e., Theorem~\ref{th:maximumArrivalLLF} remains valid for mLLF.

The goal now is to compute a schedule $\sigma$ such as $\forall t$, the budget remains between $0$ and the maximum value, i.e., $L_M=b_{\max}^{TT}$, to enforce the ET deadlines. We define two helper functions: i) last\_deadline($\mathcal{T}^{TT}$, $T$) returns the last deadline of a task within the hyperperiod $T$. This is the last theoretical slot that can be attributed to a TT task. After that time, the budget will only increase, which gives us a minimal value for the budget at the end of the hyperperiod; ii) $LL(t, \mathcal{T}^{TT})$: returns the TT task with least laxity and remaining computation time out of all ready TT tasks at time $t$.

At time $0$, the current budget is set to $L_M$. However, at the end of the hyperperiod $T$, the budget $bdg_{\sigma}(T)$ may be lower than at the start of the hyperperiod, meaning that we may not be able to repeat the exact same schedule as in the first hyperperiod. Hence, we study different hyperperiods by varying the current budget at time 0 to construct different $\sigma$. We must find a $\sigma$ fulfilling the necessary condition: $bdg_\sigma(0) \leqslant bdg_{\sigma}(T)$ to ensure that this $\sigma$ is valid $\forall t$. We also define two sufficient conditions for the schedulability and non-schedulability:
\begin{itemize}
    \item if a schedule $\sigma$ is found with the initial budget $bdg_{\sigma}(0)$ at the minimal final value $min_{\sigma_i} bdg_{\sigma_i}(T)$, i.e., corresponding to the number of idle times between the last deadline and $T$, then this schedule is valid $\forall$ t;
    \item if no schedule is found with a budget at time 0 at $bdg_{\sigma}(0)=L_M$, then no schedule exists.
\end{itemize}
\begin{algorithm}[t!]
\caption{Scheduling TT tasks under the burst limiting constraint}\label{alg:schedTTglobal}
\label{alg:gobalScheduler}
\KwData{TT tasks $\mathcal{T}^{TT}$, TT utilization $U^{TT}$, max burst $b^{TT}_{\max}$, hyperperiod $T$, processing capacity $\lambda$}
\KwResult{$\sigma$}

$I^{TT} \gets \lambda-U^{TT} $; 
$I^{idle} \gets U^{TT} $; 
$L_M \gets b^{TT}_{\max}$\;

\Comment{Minimal value of the budget at the end of an hyperperiod}
min\_budget=$\min$(T - last\_deadline($\mathcal{T}^{TT}$, T)$\cdot I^{idle} , L_M)$\;

\Comment{We first check if a schedule can be found for the minimal budget with Algorithm~\ref{alg:schedulerBudget}}
initial\_budget $\gets$ min\_budget\label{line:initialBudgetSufficient1}\;
$\sigma$=schedule(initial\_budget, $\mathcal{T}^{TT}$, $T$,$L_M$, $I^{TT}$, $I^{idle}$)\;

\If{$\sigma \neq \emptyset$}{
    returns $\sigma$;   \Comment{A schedule has been found}
}

\If{$initial\_budget == L_M$}{
  returns $\emptyset$;   \Comment{No schedule can be found}
} 
\Comment{We check if the schedule can be found with the maximum budget using Algorithm~\ref{alg:schedulerBudget}}

initial\_budget = $L_M$\label{line:initialBudgetSufficient2}\;

$\sigma$=schedule(initial\_budget, $\mathcal{T}^{TT}$, $T$, $L_M$, $I^{TT}$, $I^{idle}$)\;

\If{$\sigma == \emptyset$}{
returns $\emptyset$;   \Comment{No schedule can be found}
}
\Comment{Finally, we compute schedules with Algorithm~\ref{alg:schedulerBudget} until we find a schedule with at least as much budget at t=T as at t=0, or fail}

\While{$\sigma \neq \emptyset$ \: $\&$ initial\_budget > $\max(min\_budget, bdg_{\sigma}(t))$}{\label{line:whileschedule}
\Comment{set initial budget for the next iteration}
initial\_budget  = $\lfloor \frac{bdg_{\sigma}(t)}{I^{TT}} \rfloor \cdot I^{TT} $\label{line:initalBudgetLoop}\;
$\sigma$=schedule($initial\_budget$, $\mathcal{T}^{TT}$, $T$, $L_M$, $I^{TT}$, $I^{idle}$)\;
}

\If{$\sigma == \emptyset \vee  initial\_budget \leqslant min\_budget$}{
returns $\emptyset$;   \Comment{No schedule can be found}
}
returns $\sigma$\;
\end{algorithm}

\begin{algorithm}[t!]
\caption{Scheduling TT tasks under the BLC depending on the initial budget}

\label{alg:schedulerBudget}
\KwData{initial\_budget, TT task set $\mathcal{T}^{TT}$, hyperperiod $T$, maximum budget $L_M$, TT slot budget $I^{TT}$, idle slot budget $I^{idle}$}
\KwResult{$\sigma$}
current\_budget $\gets$ initial\_budget\;
$\forall$ $\tau_i \in \mathcal{T}^{TT}$: $c_i \gets  C_i$; $d_i \gets D_i$\;
$t \gets 0$\; 
\While{$t < T$}{\label{line:algo2:whileLoop}
    \For{$\tau_i \in \mathcal{T}^{TT}$}{
        \If{$c_i > 0 \wedge d_i \geqslant t$}{
            return $\emptyset$; \Comment{Deadline miss!}
        }
        \If{$t \% T_i == 0$}{
            \Comment{Task release at time $t$}
            $c_i \gets C_i$\; $d_i \gets t + D_i$\; 
        }  
        \If{$c_i > 0$}{
        $L_i$ $\gets (d_i - t) - c_i$; \Comment{Compute laxity of task $\tau_i$}}
    }
    \eIf{current\_budget $< L_M - I^{idle}$}{
        $L_{IDLE} \gets \lfloor \frac{\text{current\_budget}}{I^{TT}} \rfloor I^{TT} $\;
    }{
        $L_{IDLE} \gets T$;\Comment{We make sure $\tau_{IDLE}$ has the highest laxity}
    }
    \Comment{Check if $\tau_{IDLE}$ has the least laxity out of all tasks with $c_i$ > 0}
    \eIf{$L_{idle} < L_i, \forall \tau_i \in \mathcal{T}^{TT} : c_i > 0$}{
    \Comment{Schedule idle slot}
        $\sigma$[t] $\gets$ idle\; 
        current\_budget $\gets \min(current\_budget + I^{idle}, L_M)$\;
    }{\Comment{If there is enough budget, schedule the least-laxity ready task}
    \eIf{$\big(current\_budget \geqslant I^{TT}\big)  \: \wedge \big(\big[ c_i>0, \forall i \in \mathcal{T}^{TT} \big] \neq \emptyset\big)$}{
        \Comment{Schedule least-laxity ready task}
        $\sigma$[t] $\gets \tau_i = LL(t, \mathcal{T}^{TT})$\; $c_i \gets c_i - 1$\; 
        current\_budget $\gets current\_budget - I^{TT}$\;  
    }{
    \Comment{Schedule idle slot}
        $\sigma$[t] $\gets$ idle\; 
        current\_budget $\gets \min(current\_budget + I^{idle}, L_M)$\;
    }
    }
}
\If{$\big[ c_i>0, \forall i \in \mathcal{T}^{TT} \big]$}{
\Comment{Schedule is infeasible if any TT task has $c_i>0$ at this point}
    return $\emptyset$\; 
}
returns $\sigma$\;
\end{algorithm}
Hence, in Algorithm~\ref{alg:gobalScheduler}, we start by checking both sufficient conditions (Lines~\ref{line:initialBudgetSufficient1} and~\ref{line:initialBudgetSufficient2}) using the function schedule($initial\_budget$, $\mathcal{T}^{TT}$, $T$, $L_M$, $I^{TT}$, $I^{idle}$) defined in  Algorithm~\ref{alg:schedulerBudget}, where a schedule is generated according the initial budget parameter. If the sufficient conditions are not fulfilled, we run Algorithm~\ref{alg:schedulerBudget} with different initial budgets (Line~\ref{line:whileschedule}), starting with an initial budget equal to the budget at $T$ when testing the sufficient infeasibility condition. To reduce the search, we set (Line~\ref{line:initalBudgetLoop}) the new initial budget to be a multiple of $I^{TT}$ rather than directly $bdg_{\sigma}(t)$ since, due to the budget checks, as many slots can be allocated consecutively and this reduces the search space without a significant negative impact. In over $1000$ test cases, we only saw one time where the solution with the proposed initial budget failed to find a schedule that was found otherwise. However, the optimization reduced the run time by up to $98.9\%$ in some of our test cases. The algorithm ends when a solution $\sigma$ is found (i.e. $bdg_\sigma(0)$ $\leqslant$ $bdg_\sigma(T)$), or when the final budget reaches the minimum final budget possible, $bdg_\sigma(T) \leqslant \min_{\sigma_i} bdg_{\sigma_i}(T)$, (since when no solution is found, the function of the final budget $\sigma \mapsto bdg_{\sigma}(T)$ is strictly decreasing from one iteration to the next).

We note that the fundamental relation between the B3LF and any method building on the hierarchical scheduling approach is that when $T_p$ is computed in the polling approach, the maximum load of level-i $H_i(t)$ has to consider the worst-case polling task placement, denoted with $\Delta$ (c.f. Figure 3 in~\cite{10.1145/1017753.1017772}). On the other hand, with our method, we constraint the TT slot placement to fit a feasible ET schedule, generally leading to $b_{max}^{TT} \leqslant \Delta$. The more exact EDP model~\cite{4408298} may improve the schedulability of AdvPoll but will also significantly increase the complexity of solving the server design problem. However, it may be interesting the relate $b_{max}^{TT}$ and $\Delta$ (with the extended EDP model) and maybe use $b_{max}^{TT}$ to derive the polling period and deadline more quickly, but we leave such endeavor for future work.

\subsection{Complexity analysis}
\label{sec:complexity}
We denote $n^{TT}=\card{\mathcal{T}^{TT}}$,
$n^{ET}=\card{\mathcal{T}^{ET}}$, and $t$ the number of possible schedule slots on the timeline until the schedule repeats, i.e., the schedule cycle, which is either the hyperperiod $T$ or a multiple thereof. The complexity of finding $b_{max}^{TT}$ is $\mathcal{O}(n^{TT}+n^{ET})$ due to the sum over $\mathcal{T}$ to find $U_{>p}$. The complexity of Algorithm~\ref{alg:schedulerBudget} is the same as a regular LLF algorithm, namely $\mathcal{O}(n^{TT} \cdot log(n^{TT}))$ (sorting by laxity) for every time slot of a potentially exponential-length schedule cycle (c.f.~\cite{Baruah2005AlgorithmsAC}).  
The complexity of the B3LF (i.e. Algorithm~\ref{alg:schedTTglobal}) is therefore $$\mathcal{O}\bigg(\frac{ C^{TT}}{\lambda-U^{TT}}\cdot t\cdot n^{TT}\cdot log(n^{TT})\bigg)$$ because of the complexity $\mathcal{O}(\frac{ C^{TT}}{\lambda-U^{TT}}) $ of the while loop search Line~\ref{line:whileschedule} for which the schedule cycle of B3LF is upper bounded by $\lfloor\frac{L_M}{I^{TT}}\rfloor \cdot T$. We know from Theorem~\ref{th:maxTTBurst} that $L_M\leqslant  C^{TT}$ and  $I^{TT}=\lambda-U^{TT}$. Consequently, $\mathcal{O}(\frac{L_M}{I^{TT}}) \sim \mathcal{O}(\frac{ C^{TT}}{\lambda-U^{TT}})$. Hence, for our proposed method, illustrated in Figure~\ref{fig:bledf}, the complexity is $$\mathcal{O}\bigg(n^{ET}+\frac{ C^{TT}}{\lambda-U^{TT}}\cdot t\cdot n^{TT}\cdot \log(n^{TT})\bigg).$$ 

The hyperperiod length can grow exponentially as a function of the maximum period and the number of tasks $n^{TT}$~\cite{Goossens01limitationof}. However, in many practical systems, the periods of TT tasks are relatively harmonic, leading to a manageable hyperperiod length (c.f.~\cite{7176034, kramer2015real}). For example, in a real-world use-case from the automotive domain which motivated this work, the $151$ TT tasks have periods in the set $\{ 5, 10, 20, 40, 80\} ms$. Even without harmonic periods, we note that all methods, including SPoll and AdvPoll, that need to produce a static TT schedule have an intrinsic exponential component in the length of the schedule cycle. However, the methods do differ in the additional complexity of guaranteeing ET tasks. SPoll guarantees ET tasks by construction in constant time, but the resulting period may lead to an even quicker hyperperiod explosion, which can be mitigated by more aggressive (and therefore wasteful) oversampling to maintain a manageable cycle size. AdvPoll has an additional complexity in checking ET task deadlines via the response-time analysis for every server configuration and, depending on the implementation, the complexity of solving the server design problem within the TT schedulability space. B3LF, on the other hand, has only a linear additional complexity (of computing the affine envelope) bounded by $\mathcal{O}(n^{TT}+n^{ET})$.

\subsection{Design optimization}
\label{sec:design_optimization}
In the design phase of a system, it is common that not all tasks are known from the beginning, and it usually takes several iterations before the final task set is defined. Other times, some tasks might be added or changed later on in the project life-cycle. If TT tasks are known, and ET tasks are in an iterative design process, recomputing a new schedule or checking whether the old one still respects the deadlines of the new ET task set may be cumbersome.
\begin{figure}[!t]
    \centering
    \includegraphics[width=0.95\linewidth]{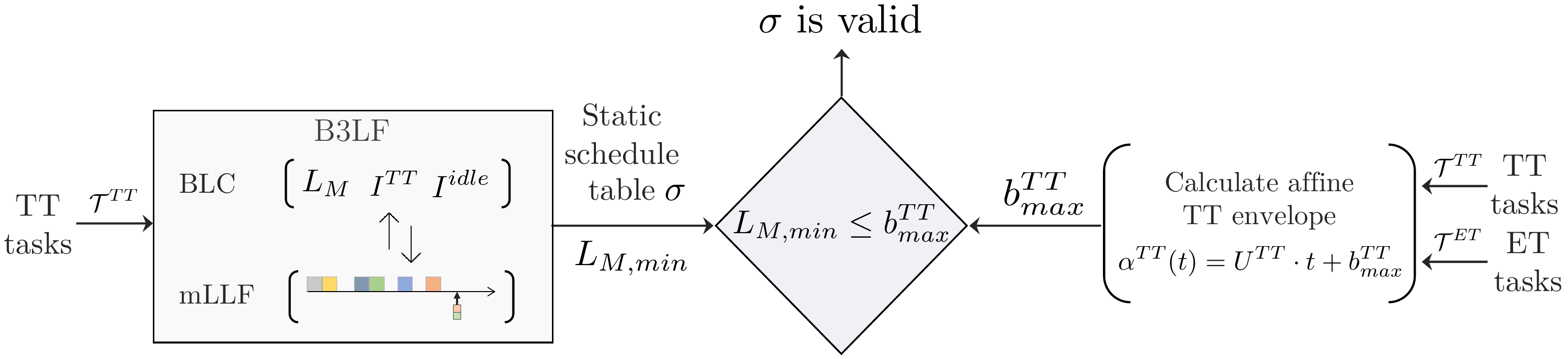}
    \caption{Design optimization of B3LF and TT affine envelop computation}
    \label{fig:b3lfOpti}
\end{figure}

With our proposed method, we can quickly check that the current ET tasks are fulfilling their deadlines by checking the $b^{TT}_{max}$ used to compute the TT schedule. In addition to this, we propose to use a design optimization, called \textbf{BinaryB3LF}, to find the minimum $L_M$ such that a schedule with only the TT tasks as input is feasible. With this $L_{M,\min}$, the scheduler needs to be run only once for the given set of TT tasks. Then, for each iteration of ET tasks, only the check of $b_{max}^{TT}\geqslant L_{M,\min}$ is needed to assess the fulfillment of ET deadlines without modifying the TT schedule, as illustrated in Figure~\ref{fig:b3lfOpti}. We know that i) $b_{max}^{TT}$ is upper bounded by $ C^{TT}$ from Theorem~\ref{th:maxTTBurst}; ii) $L_M\geqslant I_{TT}$ to be able to schedule at least one TT slot; iii) increasing $L_M$ increases the budget available for TT in the hyperperiod $T$, so the schedulability of TT depending on $L_M$ is discontinuous: not schedulable under $L_{M,\min}$, and schedulable over $L_{M,\min}$. Hence, we propose to set $I^{TT}=\lambda-U^{TT}$ and use a binary search to find the minimum value of $L_M$ such that the TT tasks are still schedulable. The search can be limited to multiples of $I^{TT}$ to improve runtime (see reasons explained for the $initial\_budget$ in Section~\ref{sec:bl3f}). 

The complexity of BinaryB3LF is $$\mathcal{O}\bigg(K\cdot n^{ET}+\log\bigg(\dfrac{ C^{TT}-\lambda+U^{TT}}{\lambda-U^{TT}}\bigg)\frac{ C^{TT}}{\lambda-U^{TT}}\cdot t\cdot n^{TT}\log(n^{TT})\bigg),$$ with $K$ the number of iterations done for the ET tasks, and  $\mathcal{O}(\log(\frac{L_M-I^{TT}}{I^{TT}}))$ being the complexity of the binary search. Without the optimization, the complexity is $\mathcal{O}(K\cdot (n^{ET}+\frac{ C^{TT}}{\lambda-U^{TT}}\cdot t\cdot n^{TT}\log(n^{TT})))$.
We note that a similar design optimization can be achieved by the simplification of choosing $C_p = \lfloor (1-U^{TT})\cdot T_p \rfloor$ in AdvPoll described in Section~\ref{sec:polling_approach} since computing the $C_p$ in this way generates the most ``dense'' allocation for the polling task for the selected $T_p$, maximizing the probability of allowing ET tasks to be added or changed without recomputing the TT schedule.

\section{Experiments}
\label{sec:experiments}

In this section, we compare our \textbf{B3LF} algorithm, also including the design optimization (\textbf{BinaryB3LF}), against \textbf{SPoll} in terms of schedulability and runtime. We implemented the SPoll method such that the polling period of each ET task does not lead to an explosion in the hyperperiod $T$, selecting the largest value lower than the ideal polling period that leads to a new hyperperiod that is smaller than $4\cdot T$. The hyperperiod explosion would be significant even for small use-cases without this additional oversampling.  After finding the polling task(s), we use a simple LLF simulation until the hyperperiod for SPoll to generate the static TT schedule.

We extended the task set generator from~\cite{756f61c4d530428194f64e61f41a0c2a, taskgen}, to create task sets containing TT and ET tasks with a deadline-monotonic priority assignment between $0$ and $6$ for ET tasks and a priority of $7$ for TT tasks. All generated task sets are schedulable if the ET tasks are considered as TT tasks and statically scheduled. All experiments were run on an Apple MacBook M1 Pro 10-core ($3.12$ GHz) machine with 16 GB of LPDDR5 memory.

For the first set of experiments, we compare the approaches in terms of schedulability and runtime for use-cases with a $10 \mu s$ microtick, $30$ TT and $20$ ET tasks per task set with periods selected from the set $\{20, 30, 40\}ms~(T=120ms)$, and $100$ task sets per test case. For the constrained ET deadline test cases, $D_i$ is uniformly selected in the upper half of the interval $[C_i, T_i]$, and for arbitrary ET deadlines, we use $D_i \in [T_i, 5\cdot T_i]$. We hence construct a favorable scenario for SPoll to be able to compete with our approach, as a smaller laxity will lead to worse results for the two classical methods (see Figure~\ref{fig:decreasing_laxity} below). 

\begin{figure}[!t]
    \centering
    \includegraphics[width=\linewidth]{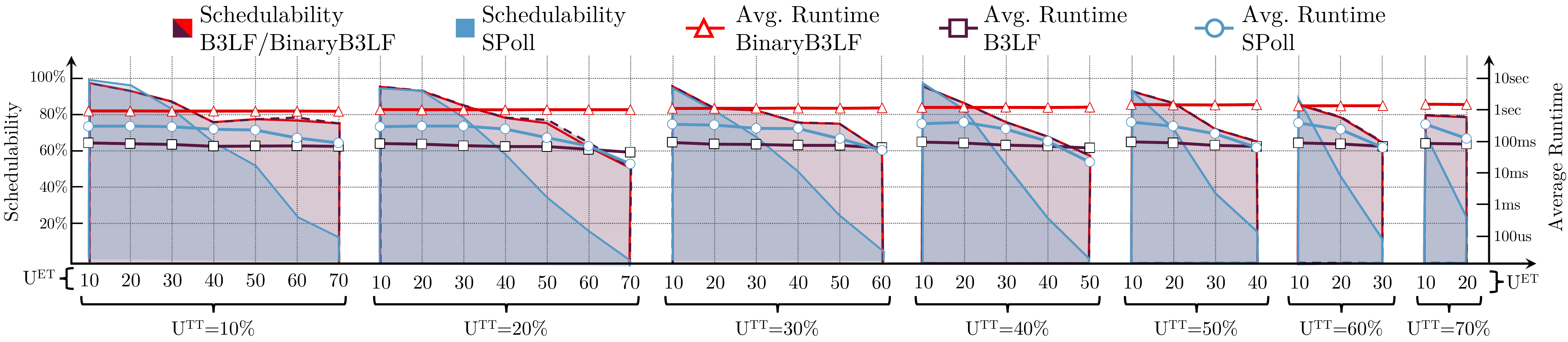}
    \caption{Schedulability and average runtime with arbitrary ET task deadlines in $[C_i, 5\cdot T_i]$.}
    \label{fig:multiplot_Di_sup}
\end{figure}
\begin{figure}[!t]
    \centering
    \includegraphics[width=\linewidth]{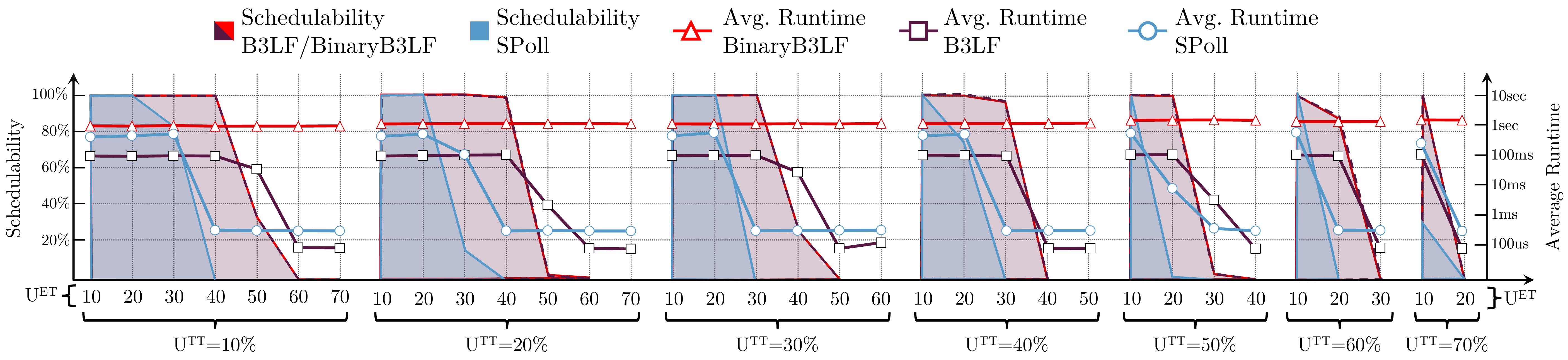}
    \caption{Schedulability and average runtime with constrained-deadline ET tasks.}
    \label{fig:multiplot_Di_inf}
    \vspace*{-10pt}
\end{figure}
We vary $u^k\in \{0.1, 0.2,..,0.7$\} such that $u^{TT}+u^{ET}\leqslant 0.9$ resulting in 34 tuples ($u^{TT}$, $u^{ET}$) as seen on the x-axis of Figures~\ref{fig:multiplot_Di_sup} and~\ref{fig:multiplot_Di_inf}. Our method consistently outperforms SPoll in terms of schedulability (left y-axis), sometimes by a significant amount, being able to schedule over $70\%$ of the task sets compared to under $10\%$ for SPoll. For arbitrary deadlines (Figure~\ref{fig:multiplot_Di_sup}), which AdvPoll cannot handle, we achieve a high test case schedulability rate even for highly utilized systems while almost always being faster than SPoll. For constrained-deadline systems (Figure~\ref{fig:multiplot_Di_inf}) with higher ET task utilization, we can still schedule relatively many task sets (sometimes even $100\%$), while the classical methods fail to schedule any. B3LF, as well as SPoll, cannot schedule any task set for systems utilizations approaching $90\%$ or for high ET task utilization ($>60\%$). For such systems, the main question is if a schedule that respects both TT and ET schedulability is possible at all. In almost all test cases, the schedulability of BinaryB3LF and B3LF is the same with a few exceptions where B3LF is better by $1-2\%$ because we consider only multiples of $I^{TT}$ for the values of $L_M$ in BinaryB3LF. 

In terms of runtime (right logarithmic y-axis), SPoll decreases with schedulability since the algorithm ends at the first polling task, for which the oversampling leads to infeasibility. The BinaryB3LF design optimization is, as expected, slower than B3LF and SPoll due to the binary search for the best $L_M$.

In the second set of experiments, we study the runtime of our approach for increasing number of tasks (Figure~\ref{fig:runtime_increasing_tasks}) and (rapidly) increasing hyperperiod (Figure~\ref{fig:runtime_increasing_hyperperiod}). Additionally, we look at the effect of decreasing laxity in constrained-deadline task sets (Figure~\ref{fig:decreasing_laxity}). 

In Figure~\ref{fig:runtime_increasing_tasks} we generate $100$ constrained-deadline task sets per test case with each task set having $20\%$ ET and $20\%$ TT task utilization, and periods chosen from the set $\{50, 100\}ms$ to keep the hyperperiod the same across test cases. We see the effect of increasing the number of tasks (x-axis) from $8$ to $1024$ (equal number of TT and ET tasks) on the runtime (logarithmic y-axis) of B3LF and BinaryB3LF, confirming the theoretical complexity findings from Section~\ref{sec:complexity} and showing the efficiency of our method in relation to SPoll. 

\begin{figure}[!t]
     \centering
     \begin{subfigure}[b]{0.31\textwidth}
         \centering
         \includegraphics[width=\textwidth]{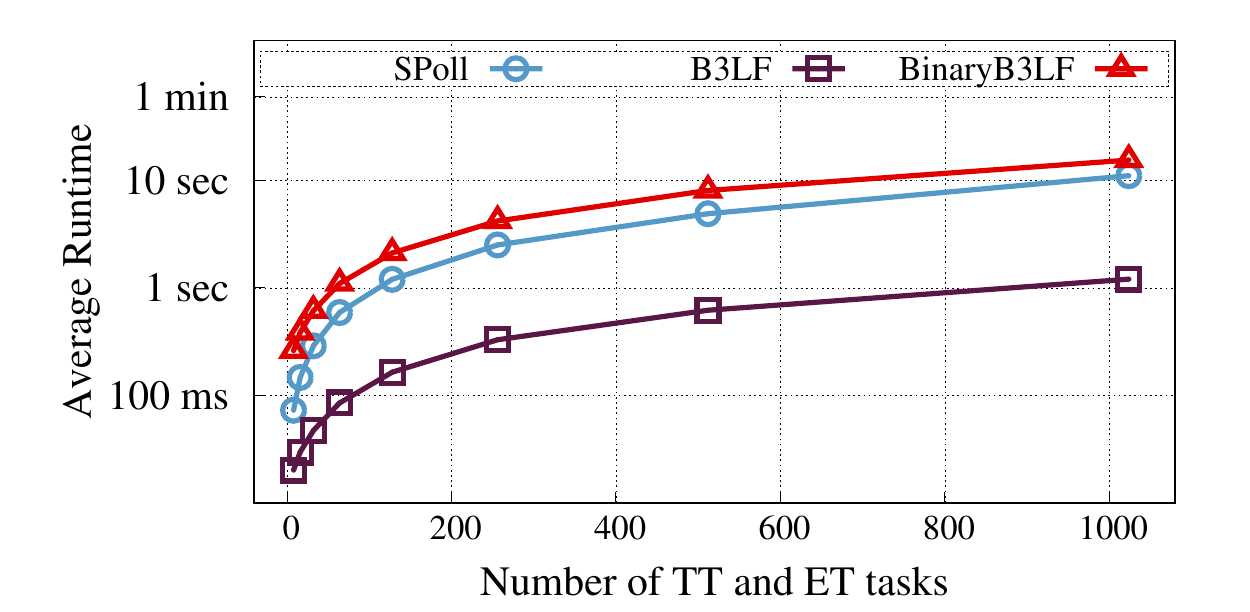}
         \caption{increasing number of tasks}
         \label{fig:runtime_increasing_tasks}
     \end{subfigure}
     \hfill
     \begin{subfigure}[b]{0.305\textwidth}
         \centering
         \includegraphics[width=\textwidth]{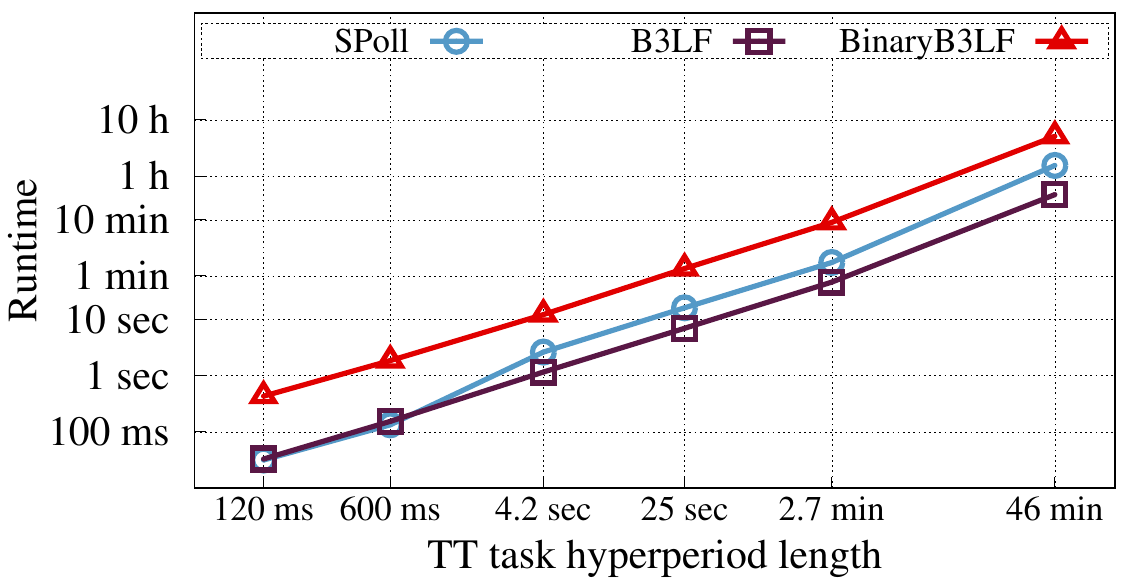}
         \caption{increasing TT hyperperiod}
         \label{fig:runtime_increasing_hyperperiod}
     \end{subfigure}
     \hfill
     \begin{subfigure}[b]{0.34\textwidth}
         \centering
         \includegraphics[width=\textwidth]{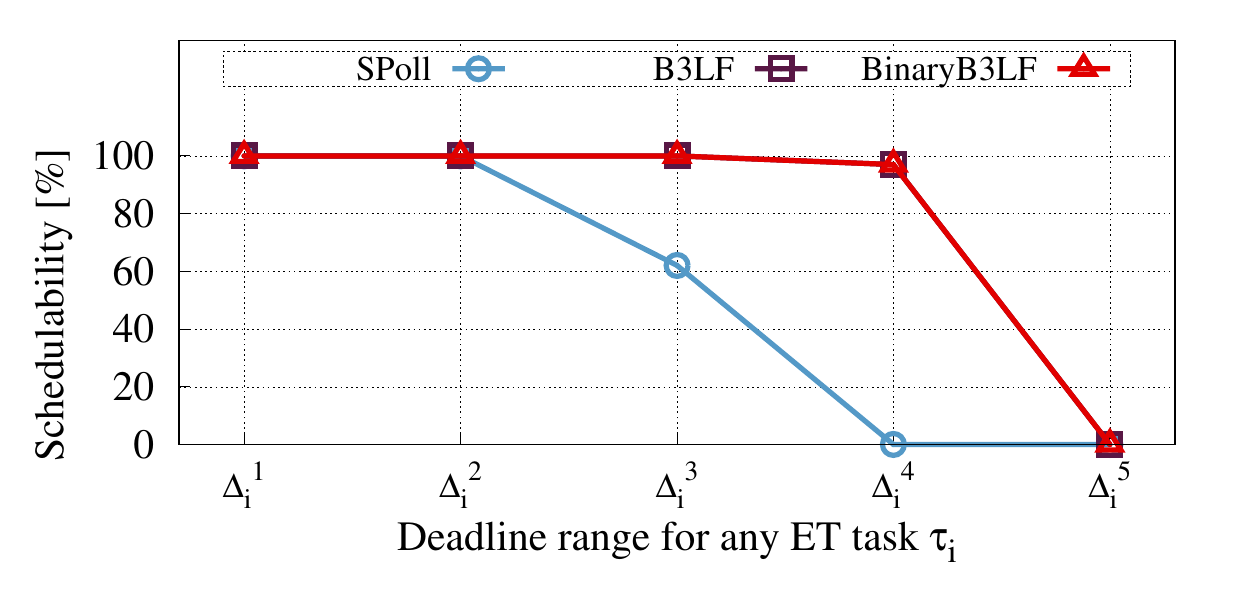}
         \caption{decreasing laxity}
         \label{fig:decreasing_laxity}
     \end{subfigure}
    \caption{Complexity and schedulability experiments for different problem dimensions.}
    \label{fig:complexity_exp}
\end{figure}

In Figure~\ref{fig:runtime_increasing_hyperperiod} we maintain the utilization setup from before and generate $1$ implicit deadline task set per test case with $8$ TT and $8$ ET tasks, increasing the hyperperiod of TT tasks $T$ exponentially from $120ms$ to $2784600ms$ ($\approx 46min$) on a timeline with a $10 \mu s$ microtick (logarithmic x-axis). The generation of the TT schedule dominates all other aspects when the hyperperiod explodes since all algorithms scale linearly in the number of time instants until their respective schedule cycles. Note that the schedule cycle is a function of the represented TT hyperperiod $T$, being either equal to it or a multiple thereof. For SPoll, the schedule cycle is the lcm of $T$ and the period(s) of the polling task(s), and for B3LF, it may be a multiple of $T$ (c.f. Algorithm~\ref{alg:schedTTglobal}) upper bounded by $\lfloor\frac{L_M}{I^{TT}}\rfloor T$. In B3LF, the computation of the maximum burst is independent of $T$. We note that even for an (unrealistically) large hyperperiod of~$\approx 46min$, B3LF manages to compute a schedule table in $27min$, which is quite acceptable for an offline schedule generation tool. 

In Figure~\ref{fig:decreasing_laxity} we maintain the setup from Figure~\ref{fig:runtime_increasing_tasks} except that each task set has $8$ TT and $8$ ET tasks. For each ET task $\tau_i \in \mathcal{T}^{ET}$, we choose the deadline randomly in each quintile of the interval $[C_i, T_i]$ in decreasing order (x-axis), i.e., $\Delta_i^k = [T_i - k(T_i-C_i)/5, T_i - (k-1)(T_i-C_i)/5], k = 1, \ldots, 5$, leading to an increasingly smaller laxity and hence making the task sets progressively harder to schedule. While for $\Delta_i^1$ and $\Delta_i^2$, the schedulability of all methods is at $100\%$, our method fares better than SPoll when the deadline of ET tasks gets more constrained for $\Delta_i^3$ and $\Delta_i^4$.  

AdvPoll is not applicable for ET tasks with arbitrary deadlines. For constrained-deadline tasks we leave the investigation of the schedulability and runtime of AdvPoll in relation to our methods for future work.

Finally, we note an interesting additional observation concerning the priority assignment of ET tasks. While we use a deadline-monotonic priority assignment (which is usual in practice), we noticed that our method's schedulability drops considerably with a random priority assignment. We found that this is due to Theorem~\ref{th:maxTTBurst}, since for fixed ET and TT task sets, when the priority $p$ decreases, $\lambda-U_{>p}$ and $-C_{\geqslant p}^{ET}$ decrease, so larger deadlines $D_j$ are needed to obtain a positive $b_{max}^{TT}$, which is not necessarily the case with random priorities.

\section{Conclusion and future work}
\label{sec:conclusion}
We have addressed the integration of sporadic event-triggered (ET) tasks with arbitrary deadlines into static time-triggered (TT) schedules via a novel method based on affine envelope approximations that distills a maximal burst and rate constraint for TT tasks such that ET tasks are schedulable. Using this affine function, we introduced an LLF-based algorithm for creating static schedule tables that respect the previously computed constraint and thereby fulfill both the temporal requirements of TT and ET tasks. We have also presented an extension that enables an efficient design optimization technique for iterative design processes where ET tasks are added or changed later. We have shown through a series of synthetic test cases that our method outperforms classical simple polling-based approaches both in terms of schedulability and runtime in most cases. 

Modern applications (e.g., in the automotive domain~\cite{McLeanETFA20, Sagstetter:2014:SIF:2593069.2593211}) are composed of multi-core multi-SoC platforms running tasks with complex dependencies (e.g., cause-effect chains~\cite{Becker:2017:ETA:3165725.3165925, 7579951}). In such distributed application scenarios, both the schedule generation and the task to core allocation are part of the scheduling challenge. In this paper, we focused on the schedule generation for individual cores and did not consider more complex dependencies between tasks. However, we note that our method has potential to be generalized and applied to distributed networked systems with complex dependencies between tasks. When generating the time-triggered schedule, our method effectively imposes a certain constraint on when slots for ET tasks must be inserted into the timeline (via the computed maximum burst). Hence, the task allocation and inter-dependence problems are, in essence, orthogonal to our method. While simple dependencies between time-triggered tasks can be readily integrated into our mLLF algorithm, adding other, more complex constraints (e.g., cause-effect chains) between tasks in distributed nodes is more challenging. We envision adding the maximum burst constraint as a special constraint on TT tasks in heuristic methods like~\cite{McLeanFRONTIERS22} and performing the maximum burst calculation for the different task-to-core allocations in swap moves of candidate solutions. Hence, we believe that our method is general enough to be applied to distributed systems with complex constraints among tasks, and we plan to investigate the integration in future work. 

\bibliographystyle{plainurl}
\bibliography{paper}

\end{document}